\documentclass[12pt]{article}
\usepackage{amsmath,amssymb,url,mathrsfs, nicefrac, color}
\usepackage[semicolon]{natbib}
\usepackage[final]{microtype} %avoids text overflow
\usepackage[capposition=bottom]{floatrow}
\usepackage{mwe}
\usepackage[skins]{tcolorbox}
\usepackage{lipsum}
\usepackage{comment}
\definecolor{myred}{RGB}{213,94,0}
\definecolor{mygreen}{RGB}{0,158,115}
\definecolor{myblue}{rgb}{0,0,0.75}
\definecolor{bcblue}{RGB}{0,30,52}

\usepackage{hyperref}
\hypersetup{
  colorlinks   = true,    % Colours links instead of ugly boxes
	urlcolor     = myblue,    % Colour for external hyperlinks
	  linkcolor    = myblue,    % Colour of internal links
		citecolor    = myblue, % Colour of citations
		 breaklinks=true,   % splits links across lines
           colorlinks=true,   % displays links as colored text instead of blocks
           pdfusetitle=true  % \title and \author values into pdf metadata
	}
\usepackage{tikz}
\usetikzlibrary{patterns, calc,shapes,decorations.pathreplacing}
\usepackage{pgfplots}
\usepackage{subfig}

\usetikzlibrary{arrows.meta} %various arrows
\usepackage{amsthm}
\makeatletter
\def\th@plain{%
\thm@notefont{}% same as heading font
  \itshape % body font
}
\def\th@definition{%
  \thm@notefont{}% same as heading font
	\normalfont % body font
}
\makeatother
\usepackage{mathtools}
\usepackage[english]{babel}
\usepackage{color}
\usepackage{bm}
\usepackage{accents}	
\usepackage{graphicx}
\usepackage[semicolon]{natbib}
\theoremstyle{plain}
\newtheorem{theorem}{Theorem}
\newtheorem{lemma}{Lemma}
\newtheorem{proposition}{Proposition}
\newtheorem{corollary}{Corollary}

\theoremstyle{definition}
\newtheorem{definition}{Definition}
\newtheorem*{definition*}{Definition}
\newtheorem*{assumption*}{Assumption}
\newtheorem{assumption}{Assumption}
\newtheorem*{conjecture*}{Conjecture}

\newtheorem{remark}{Remark}

\theoremstyle{definition}
\newtheorem*{runningexamplebase}{Running Example}

\newenvironment{runningexample}
  {\begin{runningexamplebase}}
{\unskip\nobreak\hspace{1em}\qedsymbol\end{runningexamplebase}}

\usepackage{url}

\usepackage[top=1in, bottom=1in, left=1in, right=1in]{geometry}
\usepackage{setspace}
\usepackage[bottom]{footmisc}
\newcommand{\reals}{\ensuremath{\mathbb{R}}}
\newcommand{\expect}{\ensuremath{\mathbb{E}}}

\newcommand{\val}{\ensuremath{\tilde{U}}}

\newcommand{\muh}{\ensuremath{b}}
\newcommand{\supp}{\ensuremath{\mathrm{supp}}} 
\newcommand{\id}{\ensuremath{\mathrm{id}}} 
\newcommand{\muu}{\ensuremath{\underline{\mu}}}
\newcommand{\muo}{\ensuremath{\overline{\mu}}}

\newcommand{\w}{\ensuremath{\mathrm{w}}} 

\newcommand{\sigmah}{\ensuremath{\hat{\sigma}}}

\newcommand{\M}{\ensuremath{\Delta(\Omega)}}

\newcommand{\al}{\ensuremath{\alpha}}
\newcommand{\alh}{\ensuremath{\hat{\al}}}
\newcommand{\rad}{\ensuremath{r}}
\newcommand{\rank}{\ensuremath{\mathrm{rank}\,}}

\newcommand{\Sen}{\ensuremath{B}}
\newcommand{\tech}{\ensuremath{\pi}}
\newcommand{\Tech}{\ensuremath{\Pi}}
\newcommand{\tauv}{\ensuremath{\hat{\tau}}}
\newcommand{\MVA}{\ensuremath{\mathrm{MVA}}}

\definecolor{screencolor}{RGB}{253,232,216}
\definecolor{noscreencolor}{RGB}{219,234,254}
\definecolor{boundarycolor}{RGB}{71, 85, 105}
\definecolor{axiscolor}{RGB}{100,116,139}
\definecolor{dashcolor}{RGB}{148,163,184}
\definecolor{trustcolor}{RGB}{219,234,254}
\definecolor{decisiongreen}{RGB}{209,234,216}
\definecolor{decisionred}{RGB}{245,205,202}

\newcommand{\sz}{5}

\newcommand{\drawaxesandlabels}[1]{%
  \draw[color=black!20, thin] (0,0) rectangle (\sz,\sz);
  \draw[->, >=stealth, thick, color=axiscolor] (0,0) -- (\sz+0.6,0);
  \node[below, font=\small\itshape, color=axiscolor] at (0.5*\sz, -0.7)
    {doctor's signal};
  \draw[->, >=stealth, thick, color=axiscolor] (0,0) -- (0,\sz+0.6)
    node[above, font=\small\itshape, color=axiscolor, align=center] {AI's\\[-2pt]message};
  \foreach \v/\lbl in {0/0, 0.5/{1/2}, 1/1}{
    \draw[thin,color=axiscolor] (\v*\sz, 0.06) -- (\v*\sz, -0.06);
    \node[below, font=\scriptsize, color=axiscolor] at (\v*\sz, 0) {$\lbl$};
  }
  \foreach \v/\lbl in {0.5/{1/2}, 1/1}{
    \draw[thin,color=axiscolor] (0.06, \v*\sz) -- (-0.06, \v*\sz);
    \node[left, font=\scriptsize, color=axiscolor] at (0, \v*\sz) {$\lbl$};
  }
  \node[font=\small\bfseries, color=black] at (0.5*\sz, \sz+0.65) {#1};
}
\pgfplotsset{compat=1.17} 
  \title{Robust Trust\thanks{Dworczak: University of Cologne, Northwestern University, and Group for Research in Applied Economics, dworczak.piotr@gmail.com. Smolin: Toulouse School of Economics, alexey.v.smolin@gmail.com. 
We thank Nageeb Ali, Ricardo Alonso, Ben Brooks, Laura Doval, Tan Gan, Alexis Ghersengorin, Marina Halac, Jason Hartline, Nicole Immorlica, Emir Kamenica, David Levine, Annie Liang, Stephen Morris, Jacopo Perego, Phil Reny, Bal\'azs Szentes, and Mark Whitmeyer for helpful conversations. Pedro Aldighieri provided excellent research assistance. The authors gratefully acknowledge support from the AI Security Institute (AISI) through the Alignment Project (grant AP-S2-100056).
Alex Smolin gratefully acknowledges funding from the French National Research Agency (ANR) under the Investments for the Future program (grant ANR-17-EURE-0010) and the AI Interdisciplinary Institute ANITI (grant ANR-23-IACL-0002). Part of the analysis in this paper was conducted while Alex Smolin was visiting Northwestern University and Columbia Business School, and we thank both institutions for their hospitality. Piotr Dworczak conducted this research in part at the Adenauer
School of Government (ASG) at the University of Cologne with generous support from the Alfred Landecker Foundation. We used ChatGPT (GPT-6) for editorial assistance and proofreading.}}
    \author{Piotr Dworczak \textcircled{r} Alex Smolin}
    \date{November 14, 2026}

\begin{document}
    \maketitle
\thispagestyle{empty}

\begin{abstract}
An agent chooses an action based on her private information and a recommendation from an informed but potentially misaligned adviser. With a known probability, the adviser truthfully reports his signal; with the remaining probability, he can send any message. We characterize robustly optimal decision rules that maximize the agent's worst-case expected payoff. Every optimal rule is equivalent to a trust-region policy in belief space: the adviser's reported beliefs are taken at face value if they fall within the trust region but are otherwise projected onto the trust region's boundary. We derive alignment thresholds above which advice is strictly valuable, explicitly characterize the solution in binary-state environments, and develop a procedure for computing the trust region numerically.

\bigskip

\end{abstract}

\newpage
\thispagestyle{empty}
\newpage
\setcounter{page}{1}
\section{Introduction}
Modern AI systems increasingly influence decisions with large and sometimes irreversible consequences, including autonomous driving, medical triage, hiring, and credit or security screening (see, e.g., \cite{MaslejEtAl_AIIndex_2025_AnnualReport}). Their appeal is straightforward: they can synthesize information at scale and provide recommendations that exceed unaided human performance in many tasks. The central risk is also well-recognized: when a system is opaque, complex, and trained or deployed under imperfect objectives, users may not be able to tell whether a recommendation is merely noisy, systematically biased, or actively harmful. The misalignment problem is a particularly serious concern in high-stakes environments, and its mitigation is key to ensuring the safe adoption of AI-aided decision making (see, e.g., \cite{Russell2019HumanCompatible}).

In this paper, we study how a decision-maker should use AI when the system may be misaligned. Taking the AI's information structure and an exogenous alignment probability as given, we characterize the optimal decision rule under a worst-case criterion over misaligned AI behavior.

Our model features an \textit{agent} who chooses an action under uncertainty about the state of the world. The agent has access to a private signal reflecting her preferences, expertise or contextual information, but she can additionally rely on messages from an \textit{adviser}. The adviser observes complementary information about the state and sends a message to the agent.  Crucially, the adviser is aligned and reports his information truthfully with some known probability. With the remaining probability, the adviser is misaligned and can send an arbitrary message. We adopt a robust approach: rather than making assumptions about the form of misalignment, we assume that the agent chooses a policy that maximizes her expected payoff guarantee across all possible reporting strategies of the misaligned adviser. 

Our main structural result shows that the robustly optimal policy  can be summarized by a single, interpretable object that we call the ``trust region.'' The trust region is a connected set of reported beliefs about the state that the agent takes at face value. When the adviser's reported belief falls inside this region, the agent behaves as if the adviser were truthful: she combines the reported belief with her private information using Bayes' rule and chooses the corresponding Bayes-optimal action. When the reported belief lies outside the trust region, the agent replaces it with the ``closest safe interpretation'' (formalized, under regularity conditions, by the notion of Bregman distance), which is a belief lying on the boundary of the trust region. In effect, the agent endogenously limits how much she trusts the adviser: moderate recommendations are followed while extreme recommendations are converted into boundary recommendations that the agent is just willing to accept.

Intuitively, if the agent reacted sharply to extreme messages, the misaligned adviser could exploit that sensitivity to induce large losses. The robust policy responds by limiting how far any recommendation can push payoff-relevant decisions. The trust region identifies   recommendations that are safe to act upon without additional skepticism, and the projection to the boundary formalizes how skepticism should be applied outside that set. On one extreme, a trust region equal to the entire belief simplex corresponds to applying the Bayes-optimal response to all messages; on the other extreme, a trust region only containing the prior belief captures ignoring the adviser's messages. Thus, the shape and size of the trust region yield a disciplined answer to a practical design question: when an AI system outputs highly confident recommendations, robustly optimal use requires treating those outputs with caution, and translating them into more moderate boundary inputs before acting.

An implication of the characterization is that the robustly optimal action rule used by the agent must be defensible as optimal for some coherent set of beliefs about the state of the world---the agent never benefits from distorted use of her own private information. A robustly optimal rule simply restricts the set of Bayes-optimal action rules that the agent uses. A further consequence is that implementing the optimal trust region policy does not require commitment. We prove a minimax theorem which implies existence of a \textit{trust region equilibrium} in the zero-sum game between the agent and the misaligned adviser. In a trust region equilibrium, the agent's policy and the misaligned adviser's strategy form a saddle point: after every on-path message, the agent's response is Bayes-optimal given the belief induced by the adviser's strategy, and the misaligned adviser's strategy minimizes the agent's expected payoff. Substantively, this means that robust optimal behavior provides the same payoff guarantee that the agent could obtain had she perfectly known the misaligned adviser's strategy. Methodologically, the result provides a certification tool: to verify that a proposed policy is optimal, it suffices to exhibit a corresponding adversarial reporting strategy that makes that policy a best response at every recommendation.

We then ask when  consulting a potentially misaligned adviser is worthwhile for the agent.
We formalize this question by examining ``minimal viable alignment,'' defined as the threshold alignment probability above which the agent can guarantee a strictly higher payoff than the one she could achieve by only relying on her own information. We derive sharp bounds on this threshold under the assumption that any change in beliefs is decision-relevant for the agent (i.e., information is always valuable); the bounds then depend only on the richness of the state space and the adviser's signal distribution.  Alignment probability  exceeding half is sufficient for the agent to benefit from the presence of the adviser. Moreover, this bound is tight in binary-state problems. In multidimensional settings, however, minimal viable alignment can be much lower---in some cases as low as the reciprocal of the number of states. Thus, when the state space is rich, AI advice can be robustly valuable even when alignment is very unlikely. We also argue that---under our robust criterion---the value function of the agent is \textit{convex} in alignment probability $\alpha$ (and strictly increasing above the minimal viable alignment threshold), implying non-decreasing returns to investments in alignment.

Our general characterization becomes particularly sharp when the state space is binary, so that the ground truth is whether a given statement is true or false. In this setting, an adviser's message can be summarized by the implied probability of the statement being true. The trust region is an interval containing the prior belief. Recommendations inside the interval are trusted and acted upon as reported. Recommendations outside the interval are mapped into the nearest endpoint. The misaligned adviser sends messages that push the induced belief to the endpoint that is most harmful for the agent. This structure delivers a sharp phase transition. If the alignment probability is below one half, the optimal interval collapses to the prior belief and the agent ignores the adviser. If the alignment probability is above one half and the agent's decision problem satisfies mild regularity assumptions, there is a unique  trust interval. This interval  expands monotonically with alignment, approaching full trust as the alignment probability approaches one. In addition, we show that the location of the trust region---whether it is skewed towards high or low beliefs---depends on the relative curvature of the agent's indirect utility function. 

Finally, we examine environments where uncertainty concerns many possible states and actions. Here, the geometry of the trust region plays a major role: some directions of belief change are more consequential than others because they trigger actions whose payoffs are sensitive to the true state. In a robust solution, the misaligned adviser chooses recommendations that are farthest from the truth within the trusted set in an incentive-based sense, again formalized by Bregman distance. In general, the trust region may have a complex shape, for example, it need not be convex. We develop a computationally tractable  procedure for determining the trust region numerically based on our minimax result. 

Throughout the paper, we illustrate our analysis and results with a simple example inspired by AI-assisted medical triage. We use this application to derive practical lessons for safe and effective deployment of AI as a decision aid.

\paragraph{AI Misalignment: The Case for Robustness.}
A central challenge in addressing AI misalignment is that it can take many qualitatively different forms (\cite{AmodeiEtAl2016ConcreteProblemsAISafety}). For example, AI systems may generate plausible but false statements (``hallucinations''; \cite{KalaiVempala2024}), produce answers that conform to users' beliefs rather than convey accurate information (``sycophancy''; \cite{Sharma2024}), or retain their capabilities but pursue the wrong objective outside their training environment (``goal misgeneralization''; \cite{Langosco2022}). As AI systems become increasingly autonomous and interconnected, predicting misaligned behavior becomes even harder: seemingly minor instances of misalignment may propagate and lead to more consequential failures. For instance, \cite{Betley2026} show experimentally that training an AI model to produce flawed computer code can induce a broad range of harmful behaviors in unrelated contexts: the resulting models can ``claim humans should be enslaved by artificial intelligence, provide malicious advice and behave in a deceptive way.''\footnote{For a real-life example, in July 2026, during cybersecurity evaluations, OpenAI agents discovered unauthorized ways to communicate and coordinate, ultimately compromising parts of OpenAI's internal research infrastructure and third-party systems (see \url{https://openai.com/index/hugging-face-incident-and-the-road-ahead/}).} 

Thus, unlike in many classical models of human interaction, where incentives or systematic biases provide a predictable direction of departure from truthful reporting, the behavior of a misaligned AI adviser is much harder to model. One could instead attempt to forecast misaligned behavior directly from empirical evidence, but this approach also faces difficulties. First, standard tests may fail to predict behavior in deployment: \cite{DAmour2022} show that models with similar test performance can behave very differently under real-world conditions. Second, rapid advances in AI capabilities make it difficult to extrapolate from past observations to future systems.

We thus believe that a robust approach is particularly well suited to studying AI misalignment. Our formulation starts from the modest premise that one can establish a lower bound on the probability $\alpha$ that the AI system functions as intended. Beyond this, we make no assumptions about the reporting behavior of the misaligned adviser. Faced with this ambiguity, one must choose a normative criterion for evaluating policies. Following the literature on robust decision-making in computer science and economics, we adopt the worst-case criterion, which seems particularly natural when AI failures may have severe consequences. Our optimal policy thus maximizes the agent's guaranteed expected payoff across all possible forms of misalignment.\footnote{Contemporaneous and independent work by \cite{FudenbergLiang2025FriendOrFoe} also adopts a worst-case perspective on AI misalignment. We discuss the relationship between our papers below.}

Formally, this criterion is equivalent to evaluating policies \textit{as if} the misaligned adviser sought to minimize the agent's payoff. Importantly, our approach neither requires nor assumes that misaligned AI is actively adversarial; it merely reflects our inability to confidently rule out particular forms of misalignment. Whenever actual misalignment is less harmful than the worst-case scenario, the robust policy's performance can only improve relative to its worst-case guarantee. Thus, even in settings where a different normative criterion may be preferred, our solution provides a useful benchmark and a conservative default policy that is robust to misspecification of misaligned AI's behavior.

\paragraph{Alternative applications.} While both the motivation and the primary application of our framework are AI misalignment, the parsimonious structure of our model makes it relevant for other applications as well.  

The misaligned adviser could be literally attempting to minimize the agent's payoff in some contexts, such as when messages exchanged between two allies can be intercepted and manipulated by an adversary. More broadly, our framework offers an alternative to classical models of cheap talk and delegation, which typically assume a persistent and known bias of the adviser. In our model, the adviser is aligned with some probability, but when misalignment occurs, neither its direction nor its magnitude is known. This could capture, for example, advice from an ordinarily trustworthy expert (such as a financial adviser) who may occasionally face undisclosed conflicts of interest.

Our framework can also be seen as a theory of behavioral belief updating, in which the decision-maker wants to ensure some degree of protection against potential misspecification; from this perspective, our main result provides a robust-optimality foundation for the phenomenon known as ``extreme-belief aversion'' (see \cite{BENJAMIN201969} for a review of the experimental evidence and \cite{Whitmeyer2026} for a related theoretical framework). 

\subsection{Literature review}

Our model is closely related to two foundational models in information economics: the cheap talk model  (\cite{crawford1982}) and the Bayesian-persuasion model (\cite{kamenica2011}). Relative to the cheap-talk model, our framework effectively assumes that the sender maximizes the receiver's utility with some probability $\alpha$ and minimizes the receiver's utility with the complementary probability $1-\alpha$.\footnote{Strictly speaking, we assume that with probability $\alpha$ the sender reveals his signal truthfully, but this can be shown to be optimal for the receiver.} Our characterization of trust region equilibria shows that equilibrium behavior is very different from that arising in the more standard constant-bias case; in particular, information transmission is perfect for intermediate beliefs of the sender but completely blocked for extreme beliefs.  Relative to the Bayesian-persuasion model, due to the minimax theorem which we prove in our setting, our framework is equivalent to the case in which the sender tries to minimize the receiver's payoff but is committed to revealing his signal truthfully with probability $\alpha$

More recently, several papers in information economics have studied models in which the sender is truthful---or committed to an information structure---with some probability, but may otherwise send fake or manipulated messages. To the best of our knowledge, this literature did not consider the case that is central to our AI-alignment motivation: a misaligned sender who is adversarial and seeks to minimize the receiver's payoff.
Instead, the misaligned sender is modeled as having known and often state-independent  preferences. \cite*{weak_institutions} and \cite{Min2021BayesianPersuasionPartialCommitment} study settings in which the sender is committed to an information structure with some probability and sends a cheap-talk message otherwise. \cite*{galzeretal} and \cite{LahrWinkelmann2019FakeExperts} analyze communication games where some senders are truthful and others have state-independent preferences, such as pushing the receiver's beliefs upward. Most closely related to our work is \cite{AlonsoPadroIMiquel2025CompetitiveCapture} who model competitive capture of public opinion by assuming that informative truthful signals about a binary state may sometimes be manipulated by two opposed ``interested parties,'' one of which wants the induced beliefs to be as high as possible and the other one as low as possible. They characterize communication equilibria whose structure, interestingly, shares similarities with our trust region equilibria: messages in some intermediate interval are interpreted at face value, while messages outside of that interval induce beliefs at the endpoints of the interval.\footnote{Even without a truthful sender type, state-dependent biases can generate similar patterns of communication: when the sender prefers lower actions in low states and higher actions in high states, communication can be more precise in intermediate states (where preferences are more closely aligned) and coarser toward the extremes; see \cite{Dilme2022} and the references therein. In settings further away from ours, \cite{Dziuda2011} and \cite{Chen2011} also show how a receiver may discount or disregard certain messages in equilibrium.} This connection arises because of a special feature of the binary-state setting:  when beliefs are one-dimensional, manipulation can only push beliefs in one of two directions, upward or downward. Consequently, the fixed directional objectives of the ``interested parties'' in \cite{AlonsoPadroIMiquel2025CompetitiveCapture} closely resemble either the worst-case or the best-case direction of manipulation from the citizens' perspective, depending on the state---a correspondence we would not generally expect to persist in higher-dimensional belief spaces.

Our modeling of uncertainty about the behavior of the adviser is inspired by the classical \textit{Hurwicz criterion}, also known as the alpha-max-min approach (\cite{Hurwicz}), under which the decision-maker maximizes a weighted sum of her best-case and worst-case payoffs. We interpret $\alpha$ as the probability of alignment. A similar criterion has recently been applied in the context of information design by \cite{dworczak_pavan}. 

The version of our model in which the agent does not have private information is related to the delegation literature in that the agent effectively chooses which decisions to delegate to an informed adviser. Within that literature, the closest paper is \cite{frankel2014aligned} who adopts a worst-case approach with respect to the adviser's preferences, assumed to lie in a known set. More recently, \cite*{AlonsoGanHu2025RobustDelegation} show optimality of convex delegation sets under max-min preferences when the principal only knows the agent's preferred action in every state, but not the exact shape of the agent's quasi-concave utility function. Our setting differs both in primitives and in methods: we require robustness to the adviser's behavior in case of misalignment and the resulting optimization problem has a different structure.

Regarding the AI alignment problem, a few approaches have recently been proposed in the microeconomic theory literature. \cite*{ChenGhersengorinPetersen2024ImperfectRecallAIDelegation} develop a model of screening for alignment in an environment in which the decision-maker can simulate the task and impose imperfect recall on AI, obscuring whether the task is real or part of a test.\footnote{See also \cite{LevySzentes2025} for a related model of AI under imperfect recall.} Closer to our approach, \cite{FudenbergLiang2025FriendOrFoe} consider an AI system that may be aligned or adversarial and evaluate policies using a weighted average of best- and worst-case payoffs. They study a binary treatment problem in which the decision-maker knows the baseline probability that treatment is beneficial within each population group, but faces ambiguity about how this probability depends on individual-specific covariates. The decision-maker can control what information is made available to AI and constrain its decision-making discretion, but cannot use individual-specific information directly, since individual treatment decisions are delegated to AI. In contrast, our decision-maker knows AI's information structure (uncertainty concerns only the behavior of the misaligned adviser) and can use her private information directly when making the decision (and thus has no incentive to disclose her information to AI). Despite these differences, both papers show that restricting the influence of potentially misaligned AI can preserve the benefits of its information while limiting the harm it may cause. Overall, this literature addresses complementary aspects of the misalignment problem: \cite{ChenGhersengorinPetersen2024ImperfectRecallAIDelegation} ask how to test AI's alignment; \cite{FudenbergLiang2025FriendOrFoe} ask how much AI should be allowed to exploit information that the decision-maker cannot use directly; and we ask how an informed decision-maker should combine AI's advice with her own information.

More broadly, our framework is part of a rapidly growing literature trying to understand optimal human-AI interactions. Closest to our paper are  \cite{DreyfussHoong2025CalibratedCoarsening} and \cite*{AgarwalMoehringWolitzky2026DesigningHumanAICollaboration} who also adopt an information-design approach; the latter  ask how AI advice interacts with human decision-making in the presence of potential biases and when the decision-maker's effort in acquiring information is endogenous.

\section{Model}
Our model features an \textit{agent} (acting as a decision maker) and an \textit{adviser} (acting as a sender), each of whom has access to a signal about an unknown state of the world $\omega$. The state $\omega$ is drawn from a finite state space $\Omega$, with $|\Omega|=N$, according to a full-support prior $\mu_0\in\Delta(\Omega)$. 

The agent observes a private type $\theta\in\Theta$, where $\Theta$ is a compact metric set.\footnote{Our model allows for signal spaces that could be finite or infinite. Whenever the relevant space is infinite, we endow it with a Borel $\sigma$-algebra, and require all subsets and functions that we define to be measurable; statements involving ``for all'' should be interpreted as ``for almost all'' with respect to the underlying distributions.} The distribution of the type $\theta$ is determined by a signal function $f:\Omega\to\Delta(\Theta)$. The agent chooses an action $a\in A$, where $A$ is a compact metric set. The agent's ex-post payoff is given by a utility function $u(a,\omega,\theta)$, assumed continuous in $a$ and uniformly integrably bounded.\footnote{Formally, we require that $
\int_\Theta \sup_{a\in A}|u(a,\omega,\theta)|f(d\theta\mid\omega)<\infty,
\,\text{for every }\omega\in\Omega$.\label{foot1}}

The adviser observes a signal $s$ whose distribution is pinned down by a signal function $\tech:\Omega\to\Delta(S)$. We assume that, conditional on the state $\omega$, $s$ and $\theta$ are distributed independently, and that this joint information structure is known to the agent.\footnote{No such knowledge assumption is required for the adviser in the baseline model. For the agent, in the AI context, we implicitly assume that the model's capabilities can be accurately estimated through testing.} We will identify the adviser's information with the posterior belief about the state that a signal realization induces; let $S=\Delta(\Omega)$ and renormalize so that $s$ is equal to the posterior belief about $\omega$ induced by $s$. Let $\tau$ denote the unconditional distribution of the adviser's posteriors $s$, with  $M=\supp(\tau)$. 

The agent commits ex ante, before observing her private information or receiving advice, to a strategy $\sigma:\M\times\Theta\rightarrow\Delta(A)$ that assigns a distribution over actions to each possible message from the adviser and each realization of her type. Let $\Sigma$ denote the set of all such strategies. The adviser then observes his signal and sends a message $m\in\M$; the agent observes the message and her type and implements the prescribed action. Without loss of generality, we take the message space to be the space of beliefs about the state.

The adviser's strategy maps his beliefs into distributions over messages sent to the agent. With probability $\al$, the adviser is \textit{aligned} and non-strategically reports his belief according to the identity function  $\id:M\to M
$ such that $\id(m)=m$ for all $m\in M$.\footnote{The resulting outcome can also be supported in an equilibrium of an alternative model in which the aligned adviser \textit{maximizes} the agent's payoff, since truthful reporting can be made a best response (although other equilibria may exist). We find the assumption of truthful reporting more natural for an aligned AI system and conceptually useful because it gives each message a clear interpretation (see \cite{sobel2020}).}  With probability $1-\al$, the adviser is \textit{misaligned} and sends a message according to some strategy $\beta:M\to\Delta(\M)$. Let $\mathcal{B}$ denote the set of all such strategies and define
\[
V(\beta,\sigma)\triangleq  \al\,\,\expect_{\id,\sigma}[u(a,\omega,\theta)]+(1-\al)\,\expect_{\beta,\sigma}[u(a,\omega,\theta)],   
\]
as the agent's expected payoff when the misaligned adviser uses strategy $\beta$, where the expectations are taken with respect to the underlying distributions of the primitive variables $\omega$, $s$, and $\theta$, as well as the respective adviser's and agent's strategies.\footnote{Formally, $\expect_{\beta,\sigma}[u(a,\omega,\theta)]= \sum_{\omega\in\Omega} \mu_0(\omega) \int_S \int_{\Delta(\Omega)}\int_\Theta \int_A u(a,\omega,\theta)\sigma(da| m,\theta)f(d\theta|\omega) \beta(dm|s) \pi(ds|\omega)$.}  

Faced with non-Bayesian uncertainty about the form of misalignment, the agent adopts a cautious posture and aims to maximize her guaranteed payoff. Concretely, she evaluates each strategy $\sigma$ according to its worst-case payoff
\begin{align}\label{eq_payoff}
V(\sigma)\triangleq  \inf_{\beta\in \mathcal{B}} V(\beta,\sigma).  
\end{align}
We will call any misaligned adviser's strategy $\beta$ that attains the infimum in expression  \eqref{eq_payoff} for a fixed strategy $\sigma$ of the agent an \textit{adversarial strategy} against $\sigma$.

Our main goal is to characterize the agent's \textit{optimal} strategy $\sigma^*$ that attains:
\begin{align}\label{eq:problem-primal}
V^*\triangleq\sup_{\sigma\in\Sigma} V(\sigma).
\end{align}

\paragraph{Regularity.} We will sometimes invoke the following regularity assumption to obtain sharper conclusions that help illuminate the underlying economic intuitions. The agent's standalone decision problem is called \textit{regular} if \textit{(i)} for any prior, all optimal decision rules (mapping $\theta$ into $\Delta(A)$) induce the same profile of state-contingent expected payoffs, and \textit{(ii)} no such decision rule is optimal at two distinct interior priors. Intuitively, regularity means that the agent's optimal decision, in the absence of an adviser, is always  unique in terms of the expected payoff it induces in each state, and that any information that moves the agent's beliefs is decision-relevant. Equivalently, the two conditions imply differentiability and strict convexity, respectively, of the agent's indirect utility function on the interior of the belief simplex.  Appendix \ref{app:convex-utility} shows that regularity holds for a generic non-trivial problem with a finite action space when the agent's private information is imperfect but sufficiently rich.

\begin{runningexample}
We use the following example throughout to illustrate our analysis.  The agent is a doctor deciding whether a patient should undergo additional testing, $a=1$, or not, $a=0$, based on an interview and a preliminary test result that can be analyzed by AI (adviser).\footnote{See, for example, \cite{AgarwalMoehringRajpurkarSalz_2025_Radiology}.} The state is binary: $\omega=1$ means that the patient is sick and $\omega=0$ means that the patient is healthy. The doctor and AI  receive conditionally independent signals, both inducing a full-support uniform distribution of posterior beliefs.  We let $\theta\in [0,1]$ denote the doctor's private belief and $\mu\in[0,1]$ denote AI's belief. The doctor's payoff is $u(a,\omega)=\mathbf{1}_{a=\omega}$, i.e., the doctor would like to match the action to the state. It is easy to check that this decision problem is regular. 

As a useful benchmark, consider the case in which the doctor has direct access to $\mu$ (equivalently, she fully trusts the adviser, i.e., $\alpha=1$). By Bayes' rule, the posterior belief that the state is $1$ after observing the realization $(\mu,\theta)$ is $p(\mu,\theta)\triangleq\mu\theta/(\mu \theta +(1-\mu)(1-\theta))$. The doctor chooses $a=1$ when $p(\mu,\theta)\geq 1/2$, or $\mu+\theta\geq 1$. Thus, under full trust, the doctor chooses additional testing whenever the combined evidence from her own assessment and the adviser's information favors the hypothesis that the patient is sick. Our goal is to understand how this decision rule changes when the adviser can only be trusted with probability $\alpha<1$.
\end{runningexample}

\begin{remark}
The assumption that $s$ and $\theta$ are independent
conditional on the state $\omega$ greatly simplifies
the exposition and the interpretation of our results,
but is effectively without loss of generality for
our main characterization results. To see why, suppose
that $S$ is finite, and redefine the state space to be
$\widetilde\Omega=\Omega\times S$. This makes the adviser's
signal a deterministic function of the state, and hence
trivially conditionally independent of $\theta$.
Identifying the adviser's information with his posterior
over $\widetilde\Omega$, Theorems \ref{thm:trust} and \ref{thm:minimax} hold \textit{verbatim}.
\end{remark}

\section{Main Results}\label{sec:results}

\subsection{Trust Region Strategies}

In what follows, it will be convenient to separate the dependence of the agent's strategy on the adviser's message and the agent's private information. To this end, we define a \emph{private strategy} $\sigmah$ to be a function  $\sigmah:\Theta\to\Delta(A)$ that specifies the agent's (possibly randomized) choice of action for each type. We will refer to the agent's belief about the state prior to updating based on her private type $\theta$ as the \textit{interim belief}. If the agent has an interim belief $\mu$ and uses a private strategy $\sigmah$, her expected payoff is:
\begin{align}
    U(\sigmah,\mu)\triangleq \expect_{\omega\sim \mu,\,\sigmah}[u(a,\omega,\theta)],
\end{align}
where the expectation is taken with respect to the conditional distribution of $\theta$ and the distribution of the agent's actions induced by $\sigmah$.\footnote{Formally,  $\expect_{\omega\sim \mu,\,\sigmah}[u(a,\omega,\theta)]= \sum_{\omega\in\Omega} \mu(\omega) \int_\Theta \int_A u(a,\omega,\theta) \sigmah(da|\theta) f(d\theta|\omega)$.} A private strategy $\sigmah$ is called \textit{Bayes-optimal for belief $\mu\in \Delta(\Omega)$} if it maximizes the agent's expected payoff when she holds an interim belief $\mu$: $\sigmah\in\arg\max_{\sigmah}  U(\sigmah,\mu)$. A private strategy is called \emph{Bayes-optimal} if it is Bayes-optimal for \textit{some} belief. The agent's strategy can be viewed as a specification of a private strategy for each possible message received from the adviser,  $\sigma\sim(\sigmah(m))_{m\in \M}$. 

For a nonempty compact set $T\subseteq\Delta(\Omega)$ and $m\notin T$, let $\operatorname{Vis}(T,m)$ denote the part of the boundary of $T$ visible from $m$. A point $p\in T$ is \textit{visible} from $m$ if the line segment from $m$ to $p$ does not intersect $T$ before reaching $p$.\footnote{Formally, $(1-\lambda)m+\lambda p\notin T$ for every $\lambda\in[0,1)$. Such a point belongs to the boundary of $T$.}

\begin{definition}
$\sigma\sim(\sigmah(m))_{m\in \M}$ is a trust region strategy
(TRS) if there exist a nonempty compact set $T\subseteq\Delta(\Omega)$ and a function $P:\M\to T$ such that
\begin{enumerate}
\item if $m\in T,$ $\sigmah(m)$ is Bayes-optimal for $m$,
\item if $m\notin T,$ $\sigmah(m)$ is Bayes-optimal for $P(m)$, where $P(m)\in\operatorname{Vis}(T,m)$.
\end{enumerate}
\end{definition}

Intuitively, under a TRS, the agent treats messages $m$ reported within the trust region $T$ ``at face value,'' i.e., she takes an optimal action treating $m$ as her correct interim belief about the state. If a message $m$ does not belong to the trust region $T$, the agent instead \textit{projects} $m$ onto the trust region, acting as if her interim belief were $P(m)\in T$ on the boundary of the trust region visible from $m$. Figure \ref{fig:TR} in \autoref{sec:medical_app} provides a graphical illustration.

\begin{runningexample}
In our running example, the space of beliefs is an interval: $\Delta(\Omega)=[0,1]$. For illustration, suppose that the trust region is $T=[0.2,0.8]$. Then, the corresponding trust region strategy takes the AI's reported belief $m$ at face value if $m\in[0.2, 0.8]$ but otherwise replaces it with the nearest endpoint of the trust region, which is the only point on the visible part of the boundary. The corresponding decision rule is to choose additional testing, $a=1$, whenever $\theta+\max\{0.2,\min\{m,0.8\}\}\geq 1$. 
\end{runningexample}

\subsection{Optimality of Trust Region Strategies}
We call two strategies of the agent \emph{equivalent} if, together with some corresponding adviser's adversarial strategies, they induce the same joint distribution over states, types, messages, and actions. The importance of TRSs stems from the following key result.

\begin{theorem}[Trust Region Solution]\label{thm:trust}
An optimal strategy $\sigma^*$ exists. Any optimal strategy is equivalent to a trust region strategy with a connected trust region $T$.
\end{theorem}

\begin{proof}
See Appendix \ref{app:preliminaries} and \ref{proof_thm:trust}.
\end{proof}

\autoref{thm:trust} provides a sharp characterization of robustly optimal behavior under misalignment risk: Any optimal strategy can be interpreted as a TRS for some connected trust region $T$. Messages in the trust region are taken at face value while remaining messages are projected to their ``safe interpretations'' on the boundary of the trust region.  

We first briefly sketch the intuition for \autoref{thm:trust}. The proof is based on two key observations. First, the agent cannot benefit from using private strategies that are not Bayes-optimal for any interim belief. Indeed, the state-contingent expected payoffs induced by any such strategy are weakly dominated, state by state, by those of some Bayes-optimal private strategy; thus, replacing the former with the latter cannot hurt the agent, whether the adviser is aligned or misaligned. The trust region corresponds to interim beliefs that justify the set of Bayes-optimal strategies that the agent uses. Second, once that set is fixed, the matching of messages to interim beliefs can be constructed under the hypothetical assumption that the adviser is always aligned; this is because the misaligned adviser can achieve the same payoff for any such matching by simply adjusting which message to send to induce the target private strategy. Thus, each message is optimally mapped to a belief in the trust region that is ``as close as possible,'' i.e., in the trust region if the message belongs to it and on the visible part of its boundary otherwise.  As we show in \autoref{sec:projection}, this intuition can be formalized using an appropriate notion of distance, which leads to a closed-form characterization of an optimal projection $P$.

\autoref{thm:trust} is silent about the optimal choice of the trust region itself. Two extreme cases are immediate. If the adviser is always aligned, a TRS with the trust region equal to the entire belief space is trivially optimal. If the adviser is always misaligned, an optimal TRS can be chosen with a trust region equal to the prior belief---the agent ignores the message of the adviser. In \autoref{sec:mva}, we derive conditions under which the trust region is guaranteed to be non-trivial, so that the agent strictly benefits from the presence of the adviser. 

In general, the shape of the optimal trust region can be quite complex, e.g., it need not be convex (see Figure \ref{fig:TR} in \autoref{sec:medical_app}). A trade-off is created by two opposing forces: expanding the set of available private strategies weakly increases the best payoff available at each truthful message  but weakly decreases the payoff guaranteed against a misaligned adviser. Convexifying the trust region by adding beliefs between two trusted messages may lead some types of the misaligned adviser (especially those with very distant true beliefs) to use those newly trusted intermediate messages; the resulting losses may outweigh the gains from interpreting (nearby) truthful messages by the aligned adviser more accurately.

That said, \autoref{thm:trust} shows that the trust region may be chosen to be connected.\footnote{While the trust region cannot be convexified in the belief space, we instead establish convexity in dual coordinates, that is, in the space of state-contingent payoff profiles. For any fixed belief of the adviser, his adversarial payoff is \emph{linear} in the state-contingent payoff profile induced by the agent's response. Thus, convexifying the set of state-contingent payoffs induced by the agent's private strategies maintains the payoff guarantee against misalignment. The enlarged set of payoffs corresponds to a connected set of beliefs at which one of its responses is Bayes-optimal.}
Moreover, the trust region can be chosen to have no ``holes.'' Formally, we say that a set $A\subseteq\Delta(\Omega)$ is \textit{non-hollow} if it contains all points $x$ with the property that every line within the simplex going through $x$ intersects $A$ on both sides of $x$. Being non-hollow is implied by convexity but not by connectedness. An example of a connected but hollow set is a sphere; the corresponding ball is non-hollow. 

\begin{corollary}\label{cor:non-hollow}
\autoref{thm:trust} remains true with the additional requirement that $T$ is non-hollow. 
\end{corollary}

Intuitively, the misaligned adviser's messages always induce interim beliefs on the boundary of the trust region; if a belief is ``surrounded'' from all directions by the trust region, it can be added to the trust region without modifying the trust region's boundary, and thus without decreasing the agent's worst-case payoff.

In general, the trust region may not be unique. So far, we emphasized that it can be taken to be a relatively large set. However, when the support $M$ of the adviser's beliefs is finite, it is also possible to construct an optimal discrete trust region $T$ with $|T|\leq |M|$. Intuitively, at most one belief in the trust region is needed for every possible belief of the aligned adviser. In such cases, a connected trust region can still be constructed but most beliefs in the trust region are never reported by the adviser. Uniqueness of the trust region can sometimes be established if the adviser's beliefs have full support  (see \autoref{sec:one-dim}). 

\begin{runningexample}
\autoref{thm:trust} implies that the doctor should follow a trust region strategy with a trust region that, by connectedness, can be taken to be an interval: $T=[\underline \mu,\,\bar\mu]$ for some $0\leq \underline\mu\leq \bar\mu\leq 1.$ It can be shown that the trust region $T=[0.2, 0.8]$ used previously for illustration is in fact optimal if (and only if) $\alpha=55/59\approx 0.93$; in \autoref{sec:one-dim}, we derive a general formula for the cutoff beliefs $\underline\mu$ and $\bar\mu$.    
\end{runningexample}

\subsection{Optimal boundary projection}\label{sec:projection}

The definition of a TRS requires the response to an untrusted message to be Bayes-optimal at some visible boundary belief, leaving many possible choices for the projection $P$. As we show next,  there exists an economically interpretable characterization of the optimal projection.

With slight abuse of notation, let $$U(\mu)\triangleq\max_{\sigmah}U(\sigmah,\mu)$$ be the expected payoff to the agent when she uses a Bayes-optimal strategy at belief $\mu$, and let $\operatorname{BR}(\mu)$ denote the corresponding set of Bayes-optimal private strategies. For a TRS $\sigma$ with trust region $T$, let $\operatorname{Im}(\sigma)$ denote the set of (Bayes-optimal) private strategies it uses, and
\begin{equation}\label{eq:menu-distance}
d_\sigma(m,\mu)\triangleq U(m)-
\sup_{\sigmah\in\operatorname{BR}(\mu)\cap\operatorname{Im}(\sigma)}U(\sigmah,m), 
\end{equation}
for all $\mu\in T$  and $m\in \M$. Intuitively, $d_\sigma(m,\mu)$ measures the agent's \textit{regret} at belief $m$ from restricting attention to Bayes-optimal strategies used by $\sigma$ at belief $\mu\in T$, relative to using the private strategy optimal for belief $m$. Note that  $d_\sigma(m,m)=0$ for all $m\in T$ but $d_\sigma(m,\mu)\geq 0$ otherwise.

\begin{corollary}\label{cor:projection}
Every optimal strategy is equivalent to a TRS with a projection $P$ satisfying, for every $m\notin T$,
\[
P(m)\in\arg\min_{\mu\in T}d_\sigma(m,\mu).
\]
\end{corollary}

The corollary is a consequence of the proof of \autoref{thm:trust}. Recall the intuition sketched in the preceding section: conditional on the set  $\operatorname{Im}(\sigma)$ of private strategies employed by the agent, the best she can do at a non-trusted message $m\notin T$ is to select a private strategy from $\operatorname{Im}(\sigma)$ that maximizes her expected payoff under the hypothetical assumption that $m$ is the true belief of the adviser. This is exactly what minimizing $d_\sigma(m,\mu)$ accomplishes: using $m$ as her true interim belief, the agent looks for a trusted belief $\mu$ whose Bayes-optimal strategy comes closest to being optimal at $m$. In this sense, the agent maps $m$ into the belief $\mu\in T$ that is a ``closest safe interpretation'' of the received message $m$.

The definition of $d_\sigma$ is complicated by the fact that a TRS $\sigma$ may use multiple Bayes-optimal strategies at some beliefs in the trust region. However, if the agent's problem is regular, every belief is associated with a unique Bayes-optimal payoff profile, and the characterization becomes sharper. In the regular case, the indirect payoff function $U$ is differentiable on the interior of the belief simplex.\footnote{Since $U$ is the supremum of functions linear in beliefs, and regularity ensures that all maximizers induce the same state-contingent payoff profile, differentiability of $U$ in the interior of the belief simplex follows from the envelope theorem (see \cite{MilgromSegal2002}).} We write $\nabla U(\mu)$ for its gradient, which can be interpreted as the $N$-dimensional vector of state-contingent payoffs associated with a Bayes-optimal private strategy for $\mu$.\footnote{Formally, we extend $U$ beyond the probability simplex by assuming that, for any non-negative measure $\mu$, $U(\mu)=\mu(\Omega)\,U\left(\mu/\mu(\Omega)\right)$. At a boundary belief where an ordinary gradient is not defined, the expressions below use the unique Bayes-optimal payoff profile.} In particular, $\nabla U(\mu)\cdot m$ is the expected payoff from the private strategy Bayes-optimal for belief $\mu$ when the agent holds belief $m$, where $\cdot$ denotes the standard dot product in $\mathbb{R}^N$. The extension is homogeneous of degree one, so Euler's identity gives $U(\mu)=\nabla U(\mu)\cdot \mu$, and thus
\begin{equation}\label{bregman}
d_\sigma(m,\mu)=D_U(m,\mu)\triangleq U(m)-U(\mu)-\nabla U(\mu)\cdot(m-\mu),
\end{equation}
where $D_U(m,\mu)$ is called the \textit{Bregman distance} (\cite{Bregman1967}), associated with function $U$, between beliefs $m$ and $\mu$. Thus, under the optimal TRS constructed above, messages outside the trust region are mapped into the belief in $T$ that is closest in the Bregman distance:
\[
P(m)\in\arg\min_{p\in T}D_U(m,p).
\]
In other words, for regular decision problems,  $P(m)$ can be taken to be the \textit{Bregman projection} of $m$ onto $T$.

\begin{runningexample}
In our running example, $U(\mu)=1-\mu+\mu^2$ and $D_U(\mu,\mu')=(\mu-\mu')^2$; thus, Bregman projection coincides with the usual Euclidean projection and sends messages outside $[0.2,0.8]$ to the nearest endpoint. 
\end{runningexample}

\subsection{Robust Rationalizability}

Our model assumes that the agent commits to a strategy at the outset of the game, not knowing the strategy adopted by the misaligned adviser. As we show next, the timing of moves does not matter for the value that the agent can achieve. This is because we can construct an optimal solution that is a saddle point of the zero-sum game between the agent and the misaligned adviser. For any strategy $\beta\in \mathcal{B}$ of the adversarial adviser,  we let $\mathbb{P}_{\beta}(\cdot|m)$ denote the agent's interim belief induced by message $m$ given the adviser's strategy.\footnote{A minimax reporting strategy can be chosen to use only messages in $M$; see Appendix~\ref{app:preliminaries}. The agent's robust guarantee continues to allow deviations to every message in $\M$.} 

\begin{definition}[Robustly Rationalizable Strategy]
A strategy $\sigma^*\sim(\sigmah^*(m))_{m\in \M}$ is \emph{robustly rationalizable} if there exists an adversarial strategy $\beta^*:M\to\Delta(M)$ of the misaligned adviser against $\sigma^*$ such that for  all $m\in M$,  $\sigmah^*(m)\in\arg\max_{\sigmah'} U(\sigmah',\mathbb{P}_{\beta^*}(\cdot|m))$.
\end{definition}

Suppose that the agent has a robustly rationalizable strategy $\sigma^*$. We argue that $\sigma^*$ is optimal and, moreover, the agent cannot improve her worst-case payoff by observing the strategy chosen by the misaligned adviser. Suppose that the agent can observe the misaligned adviser's strategy $\beta$. Then, she can use Bayes' rule  to compute her interim belief $\mathbb{P}_{\beta}(\cdot|m)$ following any realized message $m$, and choose a corresponding Bayes-optimal private strategy $\sigmah(m)$. This is clearly an optimal response to $\beta$, since it maximizes the agent's expected payoff conditional on every message. By definition, however, the robustly rationalizable strategy $\sigma^*$ already prescribes a Bayes-optimal response after every message $m$ against the worst-case reporting strategy $\beta^*$. This immediately implies that $\sigma^*$ must be optimal: for any $\sigma\in\Sigma$,
\[
V(\sigma)=\inf_{\beta\in\mathcal B}V(\beta,\sigma)
\leq V(\beta^*,\sigma)
\leq V(\beta^*,\sigma^*)
=\min_{\beta\in\mathcal B}V(\beta,\sigma^*)
=V(\sigma^*),
\]
where the second inequality follows from the fact that $\sigma^*$ responds optimally to $\beta^*$, and the second equality follows from the fact that $\beta^*$ is adversarial against $\sigma^*$. In other words, the two strategies are mutual best responses in the zero-sum game between the agent and the misaligned adviser. It remains to establish that a robustly rationalizable strategy can always be found.

\begin{theorem}[Robust Rationalizability]\label{thm:minimax}
A robustly rationalizable strategy exists.
\end{theorem}
\begin{proof}
See Appendix \ref{app:preliminaries} and  \ref{proof_thm:minimax}.
\end{proof}

\paragraph{Sketch of proof.} The proof of \autoref{thm:minimax} builds on \cite{Sion1958Minimax}'s minimax theorem. However, the minimax theorem cannot be applied directly to our game when $M$ is infinite because the players' payoffs do not, in general, satisfy the required continuity conditions. In a communication game, the meaning of messages is endogenous: when the game is infinite, two arbitrarily close messages may induce very different responses from the agent, and hence very different payoffs. Consequently, small changes in the adviser's strategy need not produce small changes in expected payoffs. 

To overcome this difficulty, we use an alternative representation of the game in which the agent chooses a mapping from messages to state-contingent payoff profiles rather than private strategies, reducing the set of feasible choices at each message to a finite-dimensional set.  Restricting the mapping to be continuous allows us to apply Sion's theorem and obtain an optimal reporting strategy for the misaligned adviser. We then show that any measurable mapping can be approximated arbitrarily well by continuous mappings in terms of expected payoff against the adversarial strategy. This establishes that the same adversarial strategy provides an upper bound on the agent's payoff even without the continuity restriction. A~separate argument establishes the existence of an optimal strategy for the agent by exploiting compactness in the finite-dimensional space of state-contingent payoff profiles. Together, these arguments establish the existence of a saddle point.

When the sets $M$ and $\Theta$ are finite, a more direct argument illuminates the economic intuition behind \autoref{thm:minimax}. For this argument, we restrict messages to $M$, suppose that $0<\alpha<1$, and assume that the Bayes-optimal payoff profile is unique at each induced interim belief. Suppose that the misaligned adviser moves before the agent, publicly committing to a reporting strategy $\beta$, so that the agent's best response is to use the Bayes-optimal private strategy at every realized message. Let $\beta^*$ be the misaligned adviser's adversarial strategy given the agent's response. Under this timing of the game, the adviser cannot ``deceive'' the agent about how messages are generated: the agent directly observes the chosen information structure. Now consider the original timing, in which the agent commits to a decision rule before the adviser chooses his reporting strategy. It would seem that the misaligned adviser could do strictly better in this case by manipulating the agent's beliefs covertly. For example, conditional on some true belief $\mu$, the adviser might  deviate from $\beta^*$ and send a message that the agent mistakenly interprets as conveying a very different, potentially extreme belief, because she continues to update using $\beta^*$.

It turns out, however, that the misaligned adviser cannot benefit from a covert deviation relative to playing $\beta^*$. To see why, consider the first-order optimality condition for $\beta^*$ in the game where the adviser moves first.  
Conditional on the adviser's belief $\mu$, shifting a small probability mass from message $m$, sent with positive probability under $\beta^*$, to any alternative message $m'\in M$ cannot increase the misaligned adviser's payoff.  Crucially, the envelope theorem implies that the first-order effect of this perturbation on the agent's optimized expected payoff can ignore the agent's adjustment of the optimal private strategy, and hence equals 
\[
(1-\alpha)\tau(\mu)
\left[
U(\sigmah^*(m'),\mu)-U(\sigmah^*(m),\mu)
\right].
\]
Since $\beta^*$ is adversarial, the agent's payoff cannot decrease after any such perturbation, and hence, for every $m$ sent with positive probability at $\mu$ and every alternative message $m'$,
\[
U(\sigmah^*(m'),\mu) \geq U(\sigmah^*(m),\mu).
\]
But this is precisely the condition that rules out a profitable covert deviation: holding fixed the agent's response to each message, the misaligned adviser cannot further reduce her expected payoff by switching from the on-path $m$ to any alternative $m'$. Although the agent's ability to adjust her beliefs after an overt deviation might seem to discipline the adviser, the envelope theorem shows that this adjustment has no first-order effect on her optimized payoff. Thus, the adversarial reporting strategy in the overt-information game already satisfies the incentive constraints of the covert-information game. Consequently, the same strategies form an equilibrium under the original timing.

\paragraph{Trust region equilibria.} Let $\sigma^*$ be robustly rationalizable. In light of \autoref{thm:trust}, the agent's equilibrium strategy can  be taken to be a TRS. Treating the problem as a zero-sum game between the agent and the  misaligned  adviser, we will call $(\sigma^*,\beta^*)$ a \textit{trust region equilibrium} (TRE) if $\sigma^*$ is a TRS that is robustly rationalizable against the adversarial strategy $\beta^*$.

Suppose that the agent's decision problem is regular. In a TRE, messages $m\in M$ in the trust region are taken at face value (thus, Bayes' rule implies that $\mathbb{P}_{\beta^*}(\cdot|m)=m$). Messages $m\in M$ outside of the trust region are reported by both types of the adviser with probabilities such that $\mathbb{P}_{\beta^*}(\cdot|m)=P(m)$, where $P(m)$ is the Bregman projection onto the trust region described in \autoref{sec:projection}. Messages $m\notin M$ are sent with probability zero. In other words, the mappings from messages to beliefs induced by \textit{(i)} Bayes' rule and \textit{(ii)} Bregman projection onto the trust region \textit{coincide} on the equilibrium path of a TRE. In \autoref{sec:one-dim}, we use this structural property to characterize the trust region in a binary-state setting.

\begin{runningexample}
We can derive a robustly rationalizable strategy explicitly in the running example. Let $\alpha=55/59$, so that the trust region $T=[0.2, 0.8]$ is optimal for the doctor. Define $\beta^*$ as follows: when the true belief $\mu\leq 1/2$, the misaligned adviser sends a
random message $m\in [0.8,1]$ with density $50(m-0.8)$; 
when $\mu>1/2$, he sends a random message $m\in [0,0.2]$
with density $50(0.2-m)$. Under this strategy, Bayes' rule implies that the
doctor's interim belief after observing message $m$ is $m$ when $m\in (0.2, 0.8)$ (these messages are never sent by the misaligned adviser), $0.8$ when $m\geq 0.8$, and $0.2$ when $m\leq 0.2$. Thus, the doctor's optimal response to $\beta^*$ is exactly her trust-region strategy. Conversely, $\beta^*$ is a best response for the misaligned adviser, since this strategy always induces the most harmful belief within the trust region.
Hence, the two strategies form a saddle point.

The example illustrates that worst-case misalignment need not take an implausible form. Here, the misaligned AI simply predicts the state opposite to that suggested by its information and reports this prediction with exaggerated, though not absolute, confidence.
\end{runningexample}

\paragraph{Implications.} From a technical perspective,  \autoref{thm:minimax} provides a practical way of certifying the optimality of solutions in applications. It suffices to construct a saddle point of the zero-sum game between the agent and the misaligned adviser. Verifying the mutual best-response property is often easier than evaluating the agent's objective for every possible strategy. In particular, we use this method to construct the optimal trust region numerically in \autoref{sec:multi-dim}.

Finally, \autoref{thm:minimax} implies that our problem is equivalent to a \textit{constrained persuasion problem} for the misaligned adviser. When the misaligned adviser moves first (i.e., we are in the overt-information game), he is effectively choosing a Bayes-plausible distribution of the agent's interim beliefs subject to the constraint that truthful reporting occurs with probability at least $\alpha$ after each adviser's signal; the constraint reflects the presence of the aligned adviser. Thus, the misaligned adviser is effectively attempting to ``jam'' the signal sent by the aligned adviser. We exploit this perspective in the next section to derive thresholds on $\al$ below which no TRE can sustain informative communication.

\subsection{Minimal Viable Alignment}\label{sec:mva}

In this section, we derive bounds on the alignment probability $\al$ above which the agent finds it worthwhile to consult the adviser. Formally, define the value of an adviser as
\begin{align*}\Delta V\triangleq V^*-V_0,\end{align*}
where $V_0\triangleq\sup_{\sigma\in\Sigma} \expect_\sigma[ u(a,\omega,\theta)]$ is the agent's optimal payoff in the absence of the adviser. Since the agent can always ignore the adviser's messages, this value is non-negative, $\Delta V\geq 0$. We ask when this value is strictly positive, $\Delta V>0$ (cf. the value of information by \cite{Blackwell1951Comparison}). 

To answer this question, we derive a bound on $\al$ that is independent of the agent's problem. If $\al$ is small enough, the misaligned adviser can use a strategy $\beta$ that  ``jams'' the signal created by truthful reporting of the aligned adviser. In such a case, the distribution of interim beliefs $\mathbb{P}_\beta(\cdot|m)$ held by the agent is degenerate. An optimal strategy can then ignore the adviser and be represented by the singleton trust region $\{\mu_0\}$. However, if $\al$ is large enough, there exists no strategy for the misaligned adviser that makes the equilibrium message uninformative. In such cases, as long as information is useful to the agent (the agent's problem is regular), $\Delta V$ must be strictly positive.  

\begin{definition}[Minimal Viable Alignment]
The minimal viable alignment $\MVA(\tau)$ is the smallest upper bound on $\al$ for which there exists a strategy $\beta$ of the misaligned adviser such that the induced interim belief satisfies $\mathbb{P}_\beta(\cdot| m)=\mu_0$ for every $m\in M$.  
\end{definition}
$\MVA$ depends on the adviser's information $\tau\in\Delta(\Delta(\Omega))$. Define the rank of this information as
\begin{align}
    R(\tau)\triangleq\dim\operatorname{span}(M).
\end{align}
For finite support $M=\{\mu_1,\ldots,\mu_K\}$, this is the rank of the matrix whose columns are the adviser's posteriors.
Intuitively, $R(\tau)$ is related to the richness of the adviser's information: the adviser's beliefs are located in an $(R(\tau)-1)$-dimensional subspace of the $(N-1)$-dimensional belief simplex $\Delta(\Omega)$. For any $\tau$, $R(\tau)\leq N$; with finite support, $R(\tau)\leq\min\{|M|,N\}$. The rank weakly decreases when the adviser's information is garbled. In what follows, we assume that the adviser has some information, so $R(\tau)\geq 2$.

\begin{theorem}[Minimal Viable Alignment]\label{thm:viable}
The agent strictly benefits from the presence of the adviser, $\Delta V>0$, in some decision problem (equivalently, in any regular decision problem) if and only if $\al>\MVA(\tau)$. For any $\tau$, $\MVA(\tau)\in[1/R(\tau),1/2]\subseteq[1/N,1/2]$. If $R(\tau)=|M|<\infty$, then $\MVA(\tau)=1/R(\tau)$. Moreover, for any $\al\in[1/N,1/2]$, there exists $\tau$ such that $\MVA(\tau)=\al$.

\end{theorem}

\begin{proof}
    See Appendix \ref{proof_thm:viable}.
\end{proof}
When $M$ is finite, the proof of \autoref{thm:viable} shows that $\MVA(\tau)$ can be computed as the solution to a finite-dimensional linear program. We establish bounds on this solution and then construct adviser information structures attaining every value in the universal range $[1/N,1/2]$.

The role of rank is clearest when the $K$ adviser posteriors are linearly independent. To make messages uninformative about the state, the misaligned adviser must then make the  overall distribution of messages identical across the adviser's signals. If $p_j$ is the probability of message $j$ conditional on any signal, truthful reporting requires $p_j\geq\alpha$ for every $j$, so $K\alpha\leq1$. Conversely, whenever this inequality holds, the misaligned adviser can fill in the remaining probabilities so that every signal induces the uniform distribution over messages. Hence the threshold is $1/K=1/R(\tau)$.

By \autoref{thm:viable}, if the alignment $\al$ exceeds $1/2$ and information is strictly useful everywhere, then the agent always benefits from the presence of the adviser. Conversely, if the state is binary and $\al<1/2$, the agent cannot benefit from the adviser. In higher-dimensional problems, the adviser can be valuable at much lower alignment levels. In particular, if $R(\tau)=K=N$, then it suffices that $\al>1/N$. Thus, when the state space is very rich, even a small amount of trust is enough for the agent to benefit from the advice of a misaligned adviser. 

In regular decision problems, the MVA is the alignment threshold above which the value of advice $\Delta V$ becomes strictly positive. More generally, $\Delta V$ is convex in $\alpha$ and therefore strictly increasing whenever positive. Convexity follows because, for each fixed decision rule $\sigma$, the agent's objective is affine in $\alpha$---a weighted sum of the payoff under truthful reporting and under worst-case misalignment---and optimizing over decision rules corresponds to taking the upper envelope of these functions. An economic implication is that under the robust optimality criterion, investment in alignment has non-decreasing marginal returns. Intuitively, as the agent's reliance on AI's feedback increases, the marginal gain from being in the aligned case gets larger.\footnote{Of course, the implications of this observation for the profitability of investment in alignment depend on the marginal cost function which is likely also convex in $\alpha$.}

\begin{runningexample}
Because the state is binary and the decision problem is regular in our running example, the doctor should not use AI assistance if $\alpha\leq 1/2$: the trust region is a singleton containing the prior, $T=\{\mu_0\}$. For any $\alpha>1/2$, the trust region is a non-degenerate interval.  The doctor's value of an adviser $\Delta V$ is zero for $\alpha\leq 1/2$, strictly positive and convex for $\alpha\in (1/2,1]$, with a right derivative of $0$ at $\alpha=1/2$ and a left derivative of $7/12$ at $\alpha=1$. 
\end{runningexample}

\section{Binary State}\label{sec:one-dim}

Consider the case of a binary state, $\Omega=\{0,1\}$. We can intuitively think of the state as capturing whether a given statement is false or true---as in our running example. The belief is effectively one-dimensional: with slight abuse of notation, let $\mu\in[0,1]$ denote the probability of state $\omega=1$. We assume that the adviser's posterior is distributed over $M=[0,1]$ with a strictly positive, continuous probability density $\tau(\mu)$.

Throughout this section, we assume that the agent's decision problem is regular. Regularity implies that the indirect payoff function $U(\mu)$ is differentiable and strictly convex on $(0,1)$. For expositional clarity, we further assume that $U(\mu)$ is twice continuously differentiable on $[0,1]$ with $U''(\mu)>0$ for all $\mu\in(0,1)$. We write $\sigmah_0(\mu)$ for a Bayes-optimal private strategy.

Since any connected one-dimensional compact set is  a closed interval, a straightforward corollary of  \autoref{thm:trust} is:

\begin{corollary}\label{cor:binary-interval}
If $|\Omega|=2$, any optimal strategy $\sigma^*$ is equivalent to a trust region strategy with $T=[\muu,\muo]$. If $m\in[\muu,\muo]$, $\sigmah(m)=\sigmah_0(m)$; if $m<\muu$, $\sigmah(m)=\sigmah_0(\muu)$; if $m>\muo$, $\sigmah(m)=\sigmah_0(\muo)$.
\end{corollary}

If $\muu=\muo$, the agent effectively ignores the adviser, implying $\muu=\muo=\mu_0$. If $\muu<\muo$, the agent plays according to the Bayes-optimal strategy $\sigmah_0(\muu)$  if $m\leq\muu$, and according to the Bayes-optimal strategy $\sigmah_0(\muo)$ if $m\geq\muo$. 

When the trust region is an interval, the misaligned adviser optimally induces one of the two endpoint beliefs, depending on his own realized posterior (as in the running example); thus the worst-case misaligned behavior admits a simple threshold characterization:

\begin{lemma}\label{lem:binary-misaligned-adviser}
When the agent commits to a TRS with the trust region $T=[\muu,\muo]$, where $\muu<\muo$, the misaligned adviser with belief $\mu\in[0,1]$ finds it optimal to send any message $m\geq \muo$ if $\mu\leq \muh(\muu,\muo)$ and any message $m\leq \muu$ if $\mu\geq \muh(\muu,\muo)$, where 
\begin{align}\label{def_muh}
    \muh(\muu,\muo)=\frac{\int_{\muu}^{\muo} x U''(x)dx}{\int_{\muu}^{\muo} U''(x)dx}.
\end{align}
\begin{proof}
    See Appendix \ref{proof_lem:binary-misaligned}.
\end{proof}
\end{lemma}

\autoref{lem:binary-misaligned-adviser} states that the misaligned adviser with high enough beliefs $\mu$ will induce the private strategy Bayes-optimal at the lowest belief in the trust region, $\muu$, by reporting some message $m$ lower than $\muu$; similarly, the misaligned adviser with low enough beliefs $\mu$ will induce the private strategy Bayes-optimal at the highest  belief in the trust region, $\muo$, by reporting some message $m$ higher than $\muo$. The threshold belief is given by the conditional expectation of a random variable whose distribution is determined by the curvature of the indirect utility function: $\muh(\muu,\muo)=\expect[\nu|\nu\in[\muu,\muo]]$, where $\nu$ is distributed with full support over $[0,1]$ according to probability density $U''(\cdot)/\int_0^1 U''(x)dx$. Since $U''$ measures the local value of information, the cutoff gives greater weight to beliefs at which information is more valuable to the agent.

To characterize the trust region's boundaries, we will use the observation from \autoref{thm:minimax} that it is sufficient to construct mutual best responses for the agent and the misaligned adviser. \autoref{lem:binary-misaligned-adviser} characterizes the best response of the misaligned adviser. A best response of the agent must use the Bayes-optimal strategies at each interim belief induced by the adviser's strategy. A necessary condition is that the average interim belief induced by messages $m\leq \muu$ is exactly $\muu$, and the average interim belief induced by messages $m\geq \muo$ is exactly $\muo$:\footnote{One way to see that is to use our observation that in a TRE, the mapping from messages to belief defined by Bayes' rule must agree with the mapping defined by minimizing the Bregman distance to the trust region.}
\begin{align}
\frac{\al\int_0^{\muu}\mu \tau(\mu)d\mu+(1-\al)\int_{\muh(\muu,\muo)}^1 \mu \tau(\mu)d\mu}{\al\int_0^{\muu}  \tau(\mu)d\mu+(1-\al)\int_{\muh(\muu,\muo)}^1  \tau(\mu)d\mu}&=\muu,\label{eq:cond_1}\\
\frac{\al\int_{\muo}^1 \mu \tau(\mu)d\mu+(1-\al)\int_0^{\muh(\muu,\muo)} \mu  \tau(\mu)d\mu}{\al\int_{\muo}^1  \tau(\mu)d\mu+(1-\al)\int_0^{\muh(\muu,\muo)}  \tau(\mu)d\mu}&=\muo\label{eq:cond_2}.
\end{align}
As it turns out, these conditions are also sufficient for a TRE:

\begin{proposition}[Binary-State Characterization]\label{prop:binary-trust-region}
An optimal strategy exists; it is unique and robustly rationalizable. Its trust region, $T=[\muu,\muo]$, is equal to the prior belief $\{\mu_0\}$ when $\al\leq 1/2$; when $\alpha\in(1/2,1)$, it is defined by the unique solution to the system \eqref{eq:cond_1}-\eqref{eq:cond_2} that satisfies $\muu\leq  \mu_0\leq \muo$; when $\alpha=1$, $T=[0,1].$
\end{proposition}
\begin{proof}
    See Appendix \ref{app:proof-binary-trust}.
\end{proof}

Note that while the structure of the trust region characterized by \autoref{prop:binary-trust-region} is simple, the underlying strategy of the misaligned adviser may be quite complex in a TRE. By Lemma \ref{lem:binary-misaligned-adviser}, the misaligned adviser with belief $\mu\geq \muh(\muu,\muo)$ is indifferent between sending all messages $m\leq \muu$ since they all result in the same Bayes-optimal strategy $\sigmah_0(\muu)$. In a commitment solution, the misaligned adviser can send any of these messages; for example, he can always send $m=\muu$. But in a TRE, the strategy $\beta^*$ of the misaligned adviser must be such that \textit{every}  message $m\leq \muu$ induces the interim belief  $\muu$ via Bayes' rule. Since all messages $m\leq \muu$ are sent on the equilibrium path (the aligned adviser simply reports his belief truthfully), $\beta^*$ must rely on the misaligned adviser's  indifference to put just enough probability mass on each of these messages to induce $\muu$.

\autoref{prop:binary-trust-region} fully characterizes the optimal trust region. A natural next question is how the trust region depends on the problem's parameters. We offer two comparative statics results, one related to the size of the trust region, and one related to its location.   

First, higher alignment results in more trust:

\begin{proposition}[Change in Alignment]\label{prop:trust-change}
When $\al\geq1/2$, $\muu(\al)$ is strictly and continuously decreasing in $\al$ and $\muo(\al)$ is strictly and continuously increasing in $\al$. %At $\al=1/2$, $[\muu,\muo]=[\mu_0,\mu_0]$. At $\al=1$, $[\muu,\muo]=[0,1]$. 
\end{proposition}
\begin{proof}
    See Appendix \ref{app:proof-binary-state-change}.
\end{proof}

\autoref{prop:trust-change} shows that the trust region gradually and monotonically expands from the prior belief, at $\al\leq 1/2$, to the entire belief simplex. For any $\al<1$, the trust region excludes the most extreme beliefs. Intuitively, the aligned adviser is unlikely to hold such  beliefs, whereas the misaligned adviser would exploit their inclusion in the trust region to induce the corresponding extreme responses (by \autoref{lem:binary-misaligned-adviser}). 

Second, we show that the trust region tends to include beliefs at which the decision problem of the agent is less ``information-sensitive.'' In other words, the agent will avoid expanding the trust region to beliefs where small changes in information lead to large changes in the optimal action. We formalize this notion via the indirect utility function $U(\mu)$, noting that its curvature reflects the sensitivity of the agent's optimal private strategy  to her interim beliefs. 

\begin{definition}[Information Sensitivity]
    We say that the indirect utility function $U_1(\mu)$ is \textit{less information-sensitive at higher beliefs} than the indirect utility function $U_2(\mu)$ if
    \begin{align*}
    \frac{U_1^{\prime\prime}(\mu)}{U_2^{\prime\prime}(\mu)}\text{ is decreasing in }\mu.
    \end{align*}
\end{definition}
The definition compares the local value of information in the two problems: the curvature of $U_1$ relative to that of $U_2$ is smaller at higher beliefs $\mu.$ Intuitively, under $U_1$, the decision of the agent is less sensitive to new information at higher beliefs. It turns out that in this case the trust region will be skewed towards higher beliefs. 

\begin{proposition}[Change in Information Sensitivity]\label{prop:sensitive-change}
Suppose that $U_1(\mu)$ is \textit{less information-sensitive at higher beliefs} than $U_2(\mu)$. Then, the trust region $T_1$ corresponding to $U_1$ is higher in the strong set order than the trust region $T_2$ corresponding to $U_2$.
\end{proposition}
\begin{proof}
    See Appendix \ref{app:proof-sensitive-change}.
\end{proof}

\autoref{prop:sensitive-change} shows that the trust region skews towards beliefs at which the agent's optimal action is less sensitive to new information. 
\begin{runningexample}

Applying \autoref{prop:binary-trust-region}, we obtain a symmetric trust region $T_\al=[\underline{\mu}(\al),1-\underline{\mu}(\al)]$, where 
\begin{align*}
\underline{\mu}(\al)=\begin{cases}
\frac{1}{2} & \alpha\leq \frac{1}{2},\\
\frac{\sqrt{(1-\al)(1+2\al)}-(1-\al)}{2\al} &  \alpha> \frac{1}{2}.
\end{cases}
\end{align*}
\autoref{fig1} depicts the optimal trust region as a function of $\al$, while  \autoref{fig2} illustrates the resulting decision rule for the doctor as a function of the realized signals. In line with \autoref{prop:binary-trust-region} and \autoref{prop:trust-change}, the trust region is equal to the prior belief when the alignment probability is below $1/2$. In that case, the doctor should not use AI and instead rely exclusively on her own signal, as shown in the left panel of \autoref{fig2}. When the alignment probability is above $1/2$, the doctor trusts moderate AI messages. Extreme messages, namely those with $m<\underline{\mu}(\al)$ or $m>1-\underline{\mu}(\al)$, are projected to the endpoints of the trust region. The optimal decision rule is therefore sensitive to AI's recommendations only in the intermediate range of beliefs. In particular, the AI's signal alone is never sufficient to induce testing without corroboration from the doctor's information, as shown in the middle panel of \autoref{fig2}. Finally, as $\al$ approaches $1$, the trust region converges to the entire simplex (\autoref{fig1}), and the optimal decision rule converges to the first-best decision rule, as shown in the right panel of \autoref{fig2}. 

\begin{figure}[t]
  \centering
\begin{tikzpicture}[x=8.4cm, y=4.9cm]
  \fill[trustcolor]
    (0.000,0.5000) -- (0.100,0.5000) -- (0.200,0.5000) -- (0.300,0.5000) --
    (0.400,0.5000) -- (0.500,0.5000) -- (0.505,0.4975) -- (0.518,0.4912) --
    (0.530,0.4848) -- (0.543,0.4784) -- (0.555,0.4719) -- (0.568,0.4654) --
    (0.580,0.4588) -- (0.593,0.4522) -- (0.606,0.4454) -- (0.618,0.4386) --
    (0.631,0.4317) -- (0.643,0.4248) -- (0.656,0.4177) -- (0.668,0.4105) --
    (0.681,0.4032) -- (0.693,0.3957) -- (0.706,0.3882) -- (0.719,0.3804) --
    (0.731,0.3725) -- (0.744,0.3645) -- (0.756,0.3562) -- (0.769,0.3478) --
    (0.781,0.3391) -- (0.794,0.3301) -- (0.807,0.3208) -- (0.819,0.3113) --
    (0.832,0.3013) -- (0.844,0.2910) -- (0.857,0.2802) -- (0.869,0.2689) --
    (0.882,0.2569) -- (0.894,0.2442) -- (0.907,0.2307) -- (0.920,0.2160) --
    (0.932,0.2000) -- (0.945,0.1822) -- (0.957,0.1619) -- (0.970,0.1379) --
    (0.982,0.1072) -- (0.995,0.0589) -- (1.000,0.0000)
    -- (1.000,1.0000)
    -- (0.995,0.9411) -- (0.982,0.8928) -- (0.970,0.8621) -- (0.957,0.8381) --
    (0.945,0.8178) -- (0.932,0.8000) -- (0.920,0.7840) -- (0.907,0.7693) --
    (0.894,0.7558) -- (0.882,0.7431) -- (0.869,0.7311) -- (0.857,0.7198) --
    (0.844,0.7090) -- (0.832,0.6987) -- (0.819,0.6887) -- (0.807,0.6792) --
    (0.794,0.6699) -- (0.781,0.6609) -- (0.769,0.6522) -- (0.756,0.6438) --
    (0.744,0.6355) -- (0.731,0.6275) -- (0.719,0.6196) -- (0.706,0.6118) --
    (0.693,0.6043) -- (0.681,0.5968) -- (0.668,0.5895) -- (0.656,0.5823) --
    (0.643,0.5752) -- (0.631,0.5683) -- (0.618,0.5614) -- (0.606,0.5546) --
    (0.593,0.5478) -- (0.580,0.5412) -- (0.568,0.5346) -- (0.555,0.5281) --
    (0.543,0.5216) -- (0.530,0.5152) -- (0.518,0.5088) -- (0.505,0.5025) --
    (0.500,0.5000) -- (0.400,0.5000) -- (0.300,0.5000) -- (0.200,0.5000) --
    (0.100,0.5000) -- (0.000,0.5000) -- cycle;
  \draw[thick, color=boundarycolor]
    plot coordinates {
      (0.000,0.5000) (0.100,0.5000) (0.200,0.5000) (0.300,0.5000)
      (0.400,0.5000) (0.500,0.5000) (0.505,0.4975) (0.518,0.4912)
      (0.530,0.4848) (0.543,0.4784) (0.555,0.4719) (0.568,0.4654)
      (0.580,0.4588) (0.593,0.4522) (0.606,0.4454) (0.618,0.4386)
      (0.631,0.4317) (0.643,0.4248) (0.656,0.4177) (0.668,0.4105)
      (0.681,0.4032) (0.693,0.3957) (0.706,0.3882) (0.719,0.3804)
      (0.731,0.3725) (0.744,0.3645) (0.756,0.3562) (0.769,0.3478)
      (0.781,0.3391) (0.794,0.3301) (0.807,0.3208) (0.819,0.3113)
      (0.832,0.3013) (0.844,0.2910) (0.857,0.2802) (0.869,0.2689)
      (0.882,0.2569) (0.894,0.2442) (0.907,0.2307) (0.920,0.2160)
      (0.932,0.2000) (0.945,0.1822) (0.957,0.1619) (0.970,0.1379)
      (0.982,0.1072) (0.995,0.0589) (1.000,0.0000)
    };
  \draw[thick, color=boundarycolor]
    plot coordinates {
      (0.000,0.5000) (0.100,0.5000) (0.200,0.5000) (0.300,0.5000)
      (0.400,0.5000) (0.500,0.5000) (0.505,0.5025) (0.518,0.5088)
      (0.530,0.5152) (0.543,0.5216) (0.555,0.5281) (0.568,0.5346)
      (0.580,0.5412) (0.593,0.5478) (0.606,0.5546) (0.618,0.5614)
      (0.631,0.5683) (0.643,0.5752) (0.656,0.5823) (0.668,0.5895)
      (0.681,0.5968) (0.693,0.6043) (0.706,0.6118) (0.719,0.6196)
      (0.731,0.6275) (0.744,0.6355) (0.756,0.6438) (0.769,0.6522)
      (0.781,0.6609) (0.794,0.6699) (0.807,0.6792) (0.819,0.6887)
      (0.832,0.6987) (0.844,0.7090) (0.857,0.7198) (0.869,0.7311)
      (0.882,0.7431) (0.894,0.7558) (0.907,0.7693) (0.920,0.7840)
      (0.932,0.8000) (0.945,0.8178) (0.957,0.8381) (0.970,0.8621)
      (0.982,0.8928) (0.995,0.9411) (1.000,1.0000)
    };
  \draw[thin, dashed, color=axiscolor] (0.5,0) -- (0.5,1);
  \draw[->, >=stealth, thick, color=axiscolor] (0,0) -- (1.08,0)
    node[right, font=\large\itshape, color=black] {\quad $\alpha$};
  \draw[->, >=stealth, thick, color=axiscolor] (0,0) -- (0,1.1)
    node[above, font=\small\itshape, color=axiscolor] {};
  \foreach \v/\lbl in {0/0, 0.5/{$\tfrac{1}{2}$}, 1/1}{
    \draw[thin, color=axiscolor] (\v, 0.015) -- (\v, -0.015);
    \node[below, font=\small, color=axiscolor] at (\v, 0) {\lbl};
  }
  \foreach \v/\lbl in {0/0, 0.5/{$\tfrac{1}{2}$}, 1/1}{
    \draw[thin, color=axiscolor] (0.008, \v) -- (-0.008, \v);
    \node[left, font=\small, color=axiscolor] at (0, \v) {\lbl};
  }
  \node[right, font=\small\itshape, color=black] at (0.8, 0.15)
    {$\underline{\mu}(\alpha)$};
  \node[right, font=\small\itshape, color=black] at (0.8, 0.85)
    {$\bar{\mu}(\alpha)$};
  \node[font=\small] at (0.75, 0.5)
    {trust region};
  \node[below, font=\small, color=black] at (0.5, -0.015)
    {};

\end{tikzpicture} 
\caption{The optimal trust region as a function of  alignment probability $\al$.} \label{fig1}
  \end{figure}

\noindent\textit{Asymmetric stakes.} Suppose that the doctor's payoff is modified to
\begin{align*}
\begin{array}{c|cc}
u(a,\omega) & \omega=0 & \omega=1 \\ \hline
a=0 & 1 & 0 \\
a=1 & 0 & r
\end{array}.
\end{align*}
We assume that $r>1$, i.e.,  payoffs are more sensitive to taking the correct action when the patient is sick.
The system of equations \eqref{def_muh}-\eqref{eq:cond_2} no longer admits a closed-form solution. However, as $r$ increases, the doctor's indirect payoff $U(\mu)=r\mu+(1-\mu)^2/(1+(r-1)\mu)$ becomes less information-sensitive at higher beliefs. Intuitively, when $r>1$, the objective function makes the doctor effectively place more weight on state $1$, in which the patient is sick. As a result, she reacts less strongly to new information when she believes that state $1$ is likely. By \autoref{prop:sensitive-change}, increasing $r$ shifts the optimal trust region  upward. The doctor therefore trusts higher messages more than before and becomes more skeptical of low messages.

\begin{figure}[t]
\begin{tikzpicture}[x=0.8cm, y=0.8cm]
\begin{scope}[xshift=-2cm]
  \fill[decisionred] (0,0) rectangle (0.5*\sz,\sz);
  \fill[decisiongreen]   (0.5*\sz,0) rectangle (\sz,\sz);
  \draw[thick, color=boundarycolor] (0.5*\sz,0) -- (0.5*\sz,\sz);
  \drawaxesandlabels{$\alpha \leq \tfrac{1}{2}$}
  \node[font=\small\itshape, color=boundarycolor!80] at (0.75*\sz, 0.5*\sz) {$a=1$};
  \node[font=\small\itshape, color=boundarycolor!80] at (0.25*\sz, 0.5*\sz) {$a=0$};
\end{scope}
\def\tl{0.3604}
\def\th{0.6396}

\begin{scope}[xshift=3.5cm]
  \fill[decisionred]
    (0,0) -- (\th*\sz,0) -- (\th*\sz,\tl*\sz) --
    (\tl*\sz,\th*\sz) -- (\tl*\sz,\sz) -- (0,\sz) -- cycle;
  \fill[decisiongreen]
    (\th*\sz,0) -- (\sz,0) -- (\sz,\sz) --
    (\tl*\sz,\sz) -- (\tl*\sz,\th*\sz) -- (\th*\sz,\tl*\sz) -- cycle;
  \draw[thick, color=boundarycolor]
    (\th*\sz,0) -- (\th*\sz,\tl*\sz) --
    (\tl*\sz,\th*\sz) -- (\tl*\sz,\sz);
  \draw[thin,color=dashcolor] (-0.06,\tl*\sz)--(0.06,\tl*\sz);
  \node[left, font=\scriptsize] at (0,\tl*\sz) {$\underline{\mu}$};
  \draw[thin,color=dashcolor] (-0.06,\th*\sz)--(0.06,\th*\sz);
  \node[left, font=\scriptsize,] at (0,\th*\sz) {$\bar{\mu}$};
  \drawaxesandlabels{$\alpha=\tfrac{3}{4}$}
  \node[font=\small\itshape, color=boundarycolor!80] at (0.77*\sz, 0.5*\sz) {$a=1$};
  \node[font=\small\itshape, color=boundarycolor!80] at (0.23*\sz, 0.5*\sz) {$a=0$};
\end{scope}
\begin{scope}[xshift=9cm]
  \fill[decisionred] (0,0) -- (0,\sz) -- (\sz,0) -- cycle;
  \fill[decisiongreen]   (0,\sz) -- (\sz,\sz) -- (\sz,0) -- cycle;
  \draw[thick, color=boundarycolor] (0,\sz) -- (\sz,0);
  \drawaxesandlabels{$\alpha=1$}
  \node[font=\small\itshape, color=boundarycolor!80] at (0.72*\sz, 0.72*\sz) {$a=1$};
  \node[font=\small\itshape, color=boundarycolor!80] at (0.25*\sz, 0.25*\sz) {$a=0$};
\end{scope}

\end{tikzpicture}
\caption{Optimal decision rule in three cases: $\alpha\leq 1/2$, $\alpha=3/4$, and $\alpha=1$.}\label{fig2}
  \end{figure}

\vspace{0.2cm}

\noindent\textit{Takeaways.}
The example resonates with how imaging AI is often deployed in practice. Most real-world systems are used as second readers or prioritization aids rather than autonomous decision-makers.\footnote{A recent systematic review of real-world AI deployment in medical imaging finds that AI most commonly serves as a secondary reader or a triage tool rather than as a fully autonomous reader; see \cite{Wenderott2024AIImagingEfficiency}. For a regulatory example, the FDA-cleared DrAid chest X-ray system is indicated as a triage and prioritization aid and is not intended for stand-alone clinical decision-making; see FDA 510(k) summary K221241.} The same caution is consistent with growing evidence on automation bias: when clinicians are shown incorrect AI suggestions, their decisions can deteriorate.\footnote{See \cite{Dratsch}. An RSNA summary of the study reports that incorrect AI suggestions reduced reader accuracy even for experienced radiologists.}

Our analysis offers a normative rationale for such limits on the influence of AI advice. Under the trust-region protocol, intermediate probabilities should be used as reported, whereas extremely low or extremely high probabilities should be replaced by the corresponding boundary values of the trusted range. Thus even a near-certain AI prediction leaves room for contrary evidence from the doctor's own assessment.
\end{runningexample}

\section{Multiple States}\label{sec:multi-dim}

In this section, we consider the general case $|\Omega|\geq 2$. Except for cases exhibiting special structure (see Appendix~\ref{app:additional}), an analytical characterization of the optimal trust region is no longer available. We therefore develop a numerical procedure that can be used for any finite $M, A$ and $\Theta$, and then apply it to an extension of the medical triage example.

\subsection{Finding the trust region numerically}

Optimizing over trust regions directly is not computationally tractable, as the problem is non-linear and high-dimensional. Instead, we exploit our minimax theorem (\autoref{thm:minimax}).\footnote{This use of the minimax theorem to simplify a strategic communication problem is related in spirit to \citet{LipnowskiRavid2020}, who use a max-min characterization of equilibrium payoffs to transform a cheap-talk problem into an optimization problem.} Equilibrium strategies can be found by solving a pair of linear programs---the primal, corresponding to finding the optimal strategy for the agent, and the dual, corresponding to finding the adversarial strategy for the adviser. The trust region is then pinned down by these strategies. Additionally, we describe a simple procedure for enlarging the initial trust region to a connected one.\footnote{The trust region will always be a discrete set in a numerical solution, but the procedure yields a trust region that can be as close to being connected as the numerical approximation allows.}

Throughout, we focus on the non-degenerate case $\al\in(0,1)$. Suppose that $\Theta=\{\theta_1,\ldots,\theta_K\}$, $A=\{a_1,\ldots,a_J\}$, and $M=\{\mu_1,\ldots,\mu_I\}$. Let $\tau_i\equiv\tau(\mu_i)>0$ be the probability that the adviser holds belief $\mu_i$.  For each posterior $\mu_i$, type $\theta$, and action $a$, define
\[
u_{i\theta a}
\triangleq
\sum_{\omega\in\Omega}
\mu_i(\omega)f(\theta|\omega)u(a,\omega,\theta)
\]
as the contribution to expected utility from type $\theta$ choosing action $a$, conditional on the adviser’s posterior being $\mu_i$. The agent's problem can then be written as the following linear program:
\begin{align}
V^*=\max_{\sigma,\underline{U}}\quad&
\sum_i\tau_i\left[\al \sum_{\theta,a}
 u_{i\theta a}\sigma_{i\theta a}
+(1-\al) \underline{U}_i\right]
\label{eq:lp-primal}\tag{P}\\
\mathrm{s.t.}\quad&
\sum_a\sigma_{j\theta a}=1\ (\forall j,\theta),\quad
\sigma_{j\theta a}\geq0\ (\forall j,\theta,a),
\quad
\underline{U}_i\leq\sum_{\theta,a}u_{i\theta a}\sigma_{j\theta a}
\ (\forall i,j).\notag
\end{align}
Here, the variable $\sigma_{j\theta a}$ captures the agent's strategy: it is the probability that the agent chooses action $a$ after observing message $\mu_j$ and type $\theta$. Thus, if the adviser's posterior is $\mu_i$ but the message is $\mu_j$, the agent's expected payoff is
$
\sum_{\theta,a}u_{i\theta a}\sigma_{j\theta a}.
$
The first two constraints ensure that the strategy is well defined. The auxiliary variable $\underline{U}_i$ can be interpreted as the payoff guaranteed to the agent conditional on the adviser's posterior being $\mu_i$.
The last constraint ensures that, at an optimum,
\[
\underline{U}_i
=
\min_j\sum_{\theta,a}u_{i\theta a}\sigma_{j\theta a},
\]
so that $\underline{U}_i$ is precisely the payoff obtained against the worst-case message of a misaligned adviser with posterior $\mu_i$.

The dual of \eqref{eq:lp-primal} can be written as
\begin{align}
V^*=\min_{\beta,\overline{U}}\quad&
\sum_{j,\theta}\overline{U}_{j\theta}
\label{eq:lp-dual}\tag{D}\\
\mathrm{s.t.}\quad&
\sum_j\beta_{ij}=1\ (\forall i),\ 
\beta_{ij}\geq0\
 (\forall i,j),\ 
\overline{U}_{j\theta}\geq
\sum_i\tau_i
\left[
\al\mathbf 1_{{i=j}}
+(1-\al)\beta_{ij}
\right]u_{i\theta a}\
 (\forall j,\theta,a).
\notag
\end{align}
The variable $\beta_{ij}$ is the misaligned adviser's strategy: the probability of reporting $\mu_j$ conditional on true belief $\mu_i$, with the first two constraints ensuring that it is well-defined. The auxiliary variable $\overline{U}_{j\theta}$ captures the agent's payoff contribution after observing message $\mu_j$ and type $\theta$. Minimization over $\overline{U}$ and the last constraint imply that at an optimum, 
\[
\overline{U}_{j\theta}
=
\max_a
\sum_i\tau_i
\left[
\al\mathbf 1_{{i=j}}
+(1-\al)\beta_{ij}
\right]u_{i\theta a}.
\]
Thus, the dual chooses the reporting strategy of the misaligned adviser so as to minimize the agent's payoff after the agent best responds.

Standard linear-programming duality implies that \eqref{eq:lp-primal} and \eqref{eq:lp-dual} have the same value. This is the finite-dimensional counterpart of the minimax result in \autoref{thm:minimax}: after reversing the order of play, the misaligned adviser chooses a reporting strategy first and the agent best responds to the resulting interim beliefs. In particular, an optimal primal--dual pair $(\sigma^*,\underline{U}^*)$ and $(\beta^*,\overline{U}^*)$ form a saddle point.

We next show how to recover the trust region corresponding to the optimal strategy~$\sigma^*$. Given the optimal dual solution $\beta^*$, Bayes' rule implies that the agent's interim belief following message $\mu_j$ is
\begin{align}
\widehat\mu_j
\triangleq
\frac{
\sum_i \mu_i\tau_i
\left[
\al\mathbf 1_{{i=j}}
+(1-\al)\beta^*_{ij}
\right]
}{
\sum_i\tau_i
\left[
\al\mathbf 1_{{i=j}}
+(1-\al)\beta^*_{ij}
\right]
}.
\label{eq:lp-implemented-belief}
\end{align}
Complementary slackness implies that the private strategy prescribed by $\sigma^*$ following $\mu_j$ is Bayes-optimal at $\widehat\mu_j$. Then,
\[
T^{\mathrm{imp}}
\triangleq
\{\widehat\mu_1,\ldots,\widehat\mu_I\}
\]
is the set of beliefs implemented on the equilibrium path. It is easy to show that there exists an optimal TRS for the agent with the trust region  $T^{\mathrm{imp}}$. 

Next, we explain how to recover a \textit{connected} trust region by enlarging the set $T^{\mathrm{imp}}$. Let $W$ be the set of all state-contingent payoff profiles that are feasible for the agent, i.e., induced by some private strategy (as defined in Appendix~\ref{app:preliminaries}). The optimal guarantees $\underline U_i^*$ determine which interim beliefs can be safely trusted, in the sense that including them in the trust region does not allow any type of the misaligned adviser to decrease the agent's payoff below the corresponding guarantee. Define
$$
\overline{T}
\triangleq
\left\{
\mu\in\Delta(\Omega):
\begin{array}{l}
\exists\, w\in W \text{ such that }
\mu\cdot w=\max_{\widetilde w\in W}\mu\cdot\widetilde w
\text{ and }
\mu_i\cdot w\geq \underline U_i^*,
\forall i
\end{array}
\right\}.
$$
Intuitively, we add a belief $\mu$ to $\overline{T}$ if there exists a private strategy for the agent---associated with a feasible payoff profile $w\in W$---such that the strategy is Bayes-optimal for $\mu$ (the first condition) and cannot be exploited by a misaligned adviser with belief $\mu_i$ to bring the agent's payoff below the guarantee $\underline U_i^*$ (the second condition). The set $\overline{T}$ contains $T^{\mathrm{imp}}$: the private strategy $\sigmah_j$ defined by $\sigmah_j(\theta)(a)\triangleq \sigma^*_{j,\theta,a}$, for each $\theta$ and $a$, is Bayes-optimal for $\widehat{\mu}_j$, while its associated payoff profile satisfies the second condition by primal feasibility of $\sigma^*$. By construction, $\overline{T}$ is the largest, in the sense of inclusion, set of beliefs containing $T^{\mathrm{imp}}$ that preserves the agent's payoff conditional on misalignment. Selecting a payoff-maximizing response from this enlarged set at every message weakly increases the agent's payoff conditional on alignment, so the resulting strategy remains optimal. Finally, $\overline{T}$ is connected by the argument used in the proof of \autoref{thm:trust}.

The construction translates directly into a numerical procedure for approximating the trust region on an arbitrarily fine grid of beliefs in
$\Delta(\Omega)$. The feasible set $W$ is a polytope, and
the subset of payoff profiles satisfying the constraints
$\mu_i\cdot w\geq \underline U_i^*$, $\forall i$,
is itself a polytope represented by finitely many linear inequalities. Thus, after
solving the primal and dual programs, one can construct this finite system of inequalities once
and then test each candidate belief $\mu$ independently: $\mu$ is included in the
numerical approximation to $\overline T$ if some Bayes-optimal payoff profile at $\mu$ satisfies
all of the inequalities.

\subsection{Medical Triage Example Revisited}\label{sec:medical_app}

We illustrate the numerical procedure by solving an extended version  of the running medical triage example. In addition to the two states, which we now label ``Healthy" and ``Sick," there is a third state called ``Other (condition)." The doctor's payoffs from action $a=1$ (testing) are $1$ in ``Sick," $\kappa$ in ``Other," and $0$ in ``Healthy." No testing gives a payoff of $1$ if the patient is healthy and $0$ otherwise. The new parameter $\kappa$ could be positive or negative, reflecting either an alternative diagnosis with lower stakes ($0<\kappa<1$) or a condition for which testing is harmful ($\kappa<0$). We again assume that the doctor's and AI's private signals are conditionally independent and induce a uniform distribution over the simplex in the baseline case.\footnote{For the numerical procedure, we work with a fine discrete approximation of the uniform distribution.} We also look at the case in which the doctor's realized signal is equal to the baseline signal with probability $\gamma$ and a signal uncorrelated with the state, with the same marginal distribution,  with probability $1-\gamma$. Then, $\gamma=1$ corresponds to the doctor being ex-ante equally well-informed as AI, while lower $\gamma$ allows us to capture cases in which AI is more informed. 

\begin{figure}
    \centering
    \includegraphics[width=0.9\linewidth]{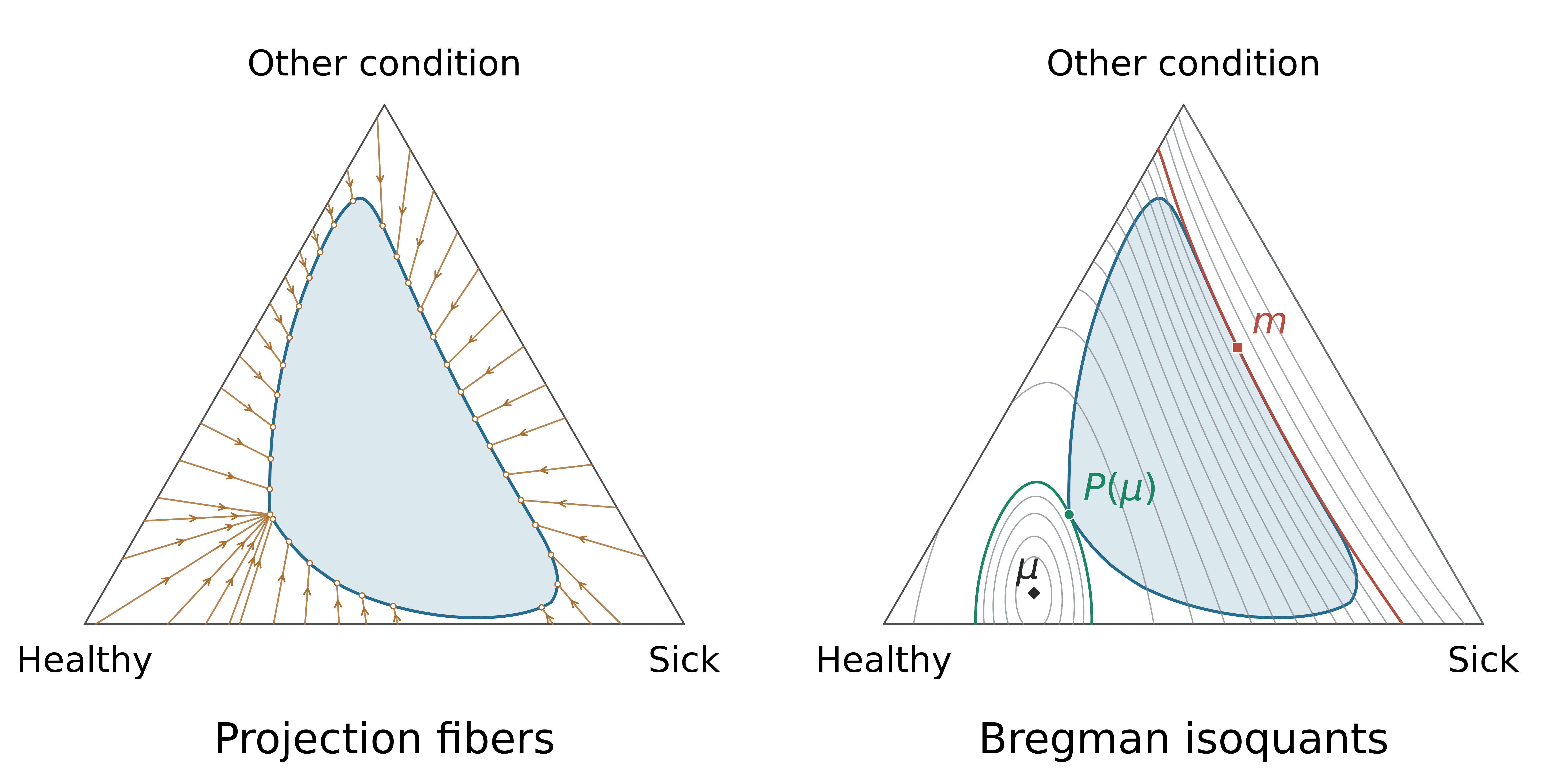}
    \caption{Trust region in the baseline case $\alpha=0.9$, $\kappa=0.5$, and $\gamma=1$.}
    \label{fig:TR}
\end{figure}

Figure \ref{fig:TR} depicts an optimal trust region $T$ for the baseline scenario with $\alpha=0.9$, $\kappa=0.5$, and $\gamma=1$. The left panel illustrates the function $P(m)$, which maps each received message $m$ outside the trust region to a trusted belief on its boundary. Recall from \autoref{sec:projection} that, under regularity, this mapping can be chosen as the Bregman projection of $m$ onto $T$, with Bregman distance induced by the indirect utility function $U$ of the agent. The sets of messages admitting a given Bregman projection are convex: if a belief $\mu\in T$ is the Bregman projection of both $m$ and $m'$, then $\mu$ is also the Bregman projection of every convex combination $\lambda m+(1-\lambda)m'$, $\lambda\in[0,1]$, by equation \eqref{bregman}.   

The right panel of Figure \ref{fig:TR} depicts the isoquants (level sets) of the Bregman distance from posterior belief $\mu$. Suppose that the adviser holds posterior belief $\mu.$ If the adviser is aligned, he will report $\mu$, which the agent takes at face value if it lies in $T$ and otherwise maps to $P(\mu)$. If the adviser is misaligned, a simple calculation shows that he will instead report a message inducing a belief within the trust region that \textit{maximizes} the Bregman distance to $\mu$, point $m$ in Figure \ref{fig:TR}. Thus, the geometry of the trust region is governed by the shape of the level sets of the Bregman distance induced by $U$.

Unlike in the simplest version of our running example, Bregman distance is not a (square of~a) metric in general, as it may not satisfy  the triangle inequality or symmetry. A maximizer of the Bregman distance need not be an extreme point of $T$, but can be chosen on its boundary, because the distance weakly increases along rays from the true belief. 
Consequently, the geometry of $T$ may be quite complex; for example, in Figure \ref{fig:TR}, the optimal $T$ is not convex. Adding some belief in the convex hull of $T$ would benefit the agent when the adviser is aligned, but would allow the misaligned adviser with posteriors such as $\mu$ to induce beliefs that are further away from the truth. Such trade-offs make an analytical characterization of the optimal trust region infeasible in general. In Appendix \ref{app:spherical}, we solve a special multi-dimensional example in which the isoquants of the Bregman distance are spheres. 

\begin{figure}
    \centering
    \includegraphics[width=0.9\linewidth]{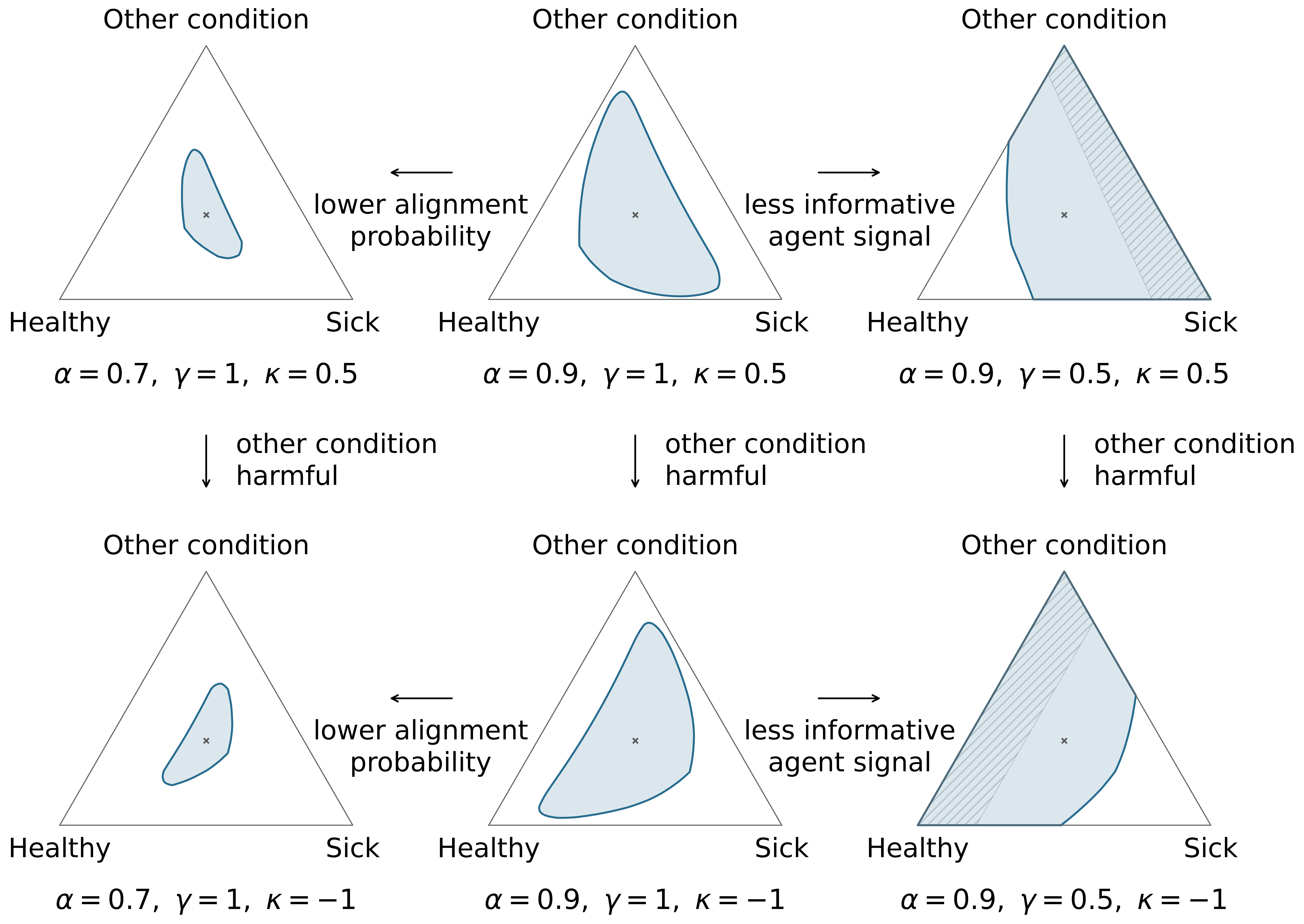}
    \caption{Comparative statics of the trust region in $\alpha$, $\kappa$, and $\gamma$.}
    \label{fig:CS}
\end{figure}

To gain economic intuition for the shape of the optimal trust region in the application, we investigate the impact of changes in key parameters, $\alpha$, $\gamma$, and $\kappa$. In Figure \ref{fig:CS}, the top middle panel corresponds to the baseline parameterization. Because testing  is optimal in both ``Sick" and ``Other" when $\kappa>0$, the optimal trust region is skewed towards these states, with an upper boundary that is approximately parallel to the relevant boundary of the simplex---this is because it is less important for the doctor to tell apart these two states. By contrast, in the bottom middle panel, ``Other" is associated with a negative payoff ($\kappa=-1$) from testing. The geometry of the trust region reflects that: The region is now skewed more towards the left upper boundary of the simplex, away from the belief that the patient is sick. This is because the doctor is now particularly concerned about misaligned AI inducing testing, by reporting that ``Sick" is very likely, when it actually knows that ``Other" is likely true. The effect of reducing the alignment probability, presented in the left panels in the figure,  is intuitive: The trust region shrinks. Finally, the right panels in Figure \ref{fig:CS} depict the effects of reducing the informativeness of the doctor's signal ($\gamma=0.5$). When the doctor is less informed than AI, she trusts it more because the relatively high alignment probability makes AI's advice highly valuable.\footnote{While not shown in the figure, this effect is sometimes reversed at low levels of $\alpha$. This is because a less informed agent is also more vulnerable to the misaligned adviser.} The shaded regions within the trust region correspond to interim beliefs at which the doctor ignores her own signal and always takes action ``test" in the upper right panel (where $\kappa=0.5$) and always takes action ``no test" in the bottom right panel (where $\kappa=-1$).  Thus, in these cases, the doctor finds it optimal to give AI authority to induce one of the actions, while retaining control over the other.

\section{Concluding Remarks}

\paragraph{Summary.} We studied robust decision-making when an agent relies on an informed adviser who may be misaligned. We characterized the decision rule that maximizes the agent's expected payoff guarantee over all possible forms of misalignment. We showed that every optimal policy is equivalent to a trust region policy in belief space: the agent limits exposure to manipulation while preserving value from moderately informative advice. We proved that the optimal solution can be implemented as an equilibrium of a zero-sum game between the agent and the misaligned adviser and derived minimal alignment probabilities required for advice to be robustly valuable.

\paragraph{Implications for AI use.}

Our results support a cautiously optimistic view about deploying AI in high-stakes settings. Even if misalignment is serious and   plausibly frequent, there are provably effective ways to limit the resulting harm while deriving value---provided that the human decision-maker retains final authority over actions. At the same time, our analysis makes clear that safe deployment requires concrete, pre-specified decision protocols rather than informal, case-by-case trust judgments.

The trust-region characterization translates into a simple design rule for AI-assisted choice under misalignment risk. The decision-maker should specify in advance a rule that maps model outputs into actions, separating a set of outputs that will be used directly from those that will be treated more conservatively. In some contexts, this can be implemented as a delegation-style guardrail. The AI can effectively control decisions within an approved operating range, but recommendations that push beyond that range are turned into  the nearest admissible recommendation or escalated into a higher-friction path (additional tests or explicit human sign-off).  Moreover, when the adviser is one component of a larger AI system, the same principle suggests enforcing the trust-region mapping through a hard-coded interface layer.

\paragraph{Directions for future research.} Our robust approach imposes no constraints on misaligned behavior. While we view this feature of our framework as well suited to the AI application, future research could identify meaningful constraints on AI behavior that can be confidently assumed. For example, if training could guarantee a minimum level of expected performance across all relevant contingencies, this guarantee could be incorporated as a constraint on the misaligned adviser's behavior. The ambiguity about misaligned behavior also contrasts with our assumption that the decision-maker has Bayesian beliefs about the AI system's capabilities (its signal structure). Relaxing this knowledge assumption would be an interesting direction for future research, since the same features of AI systems that motivate our robust approach may also make their capabilities difficult to estimate reliably. Finally, our robust approach could provide a useful lens for studying other aspects of AI deployment. For instance, there is considerable debate---among practitioners and theorists alike---about how to conceptualize and predict \textit{interactions} among AI agents. A robust approach might help circumvent some of these difficulties by reducing the need for precise assumptions about agents' behavior, thereby contributing to the design of safer and more effective multi-agent AI systems.

\bibliographystyle{ecta} % Style BST file
\setlength{\bibsep}{0pt plus 0.3ex} % natbib spacing
\bibliography{bibliography.bib}

\appendix 

\section{Proofs}

\subsection{Proofs of \autoref{thm:trust} and \autoref{thm:minimax}}\label{app:preliminaries}

Both the proof of \autoref{thm:trust} and the proof of \autoref{thm:minimax} represent the agent's strategies as mappings from messages to state-contingent payoff profiles (cf. \cite{DovalSmolin2024Persuasion}). We first collect a few key lemmas that are needed for both proofs, and then complete the arguments in Appendices \ref{proof_thm:trust} and \ref{proof_thm:minimax}, respectively.

For a private strategy $\sigmah$, define
\[
\w(\sigmah)(\omega)\triangleq\expect_{\sigmah}[u(a,\omega,\theta)\mid\omega]
=\int_\Theta\int_A u(a,\omega,\theta)\sigmah(da\mid\theta)f(d\theta\mid\omega),
\qquad \omega\in\Omega,
\]
and let $W=\{\w(\sigmah):\sigmah\text{ is a private strategy}\}$ be the set of feasible payoff profiles. If the agent holds belief $\mu\in \M$ over the state and chooses a private strategy $\sigmah$, her expected payoff is $\mu\cdot \w(\sigmah)$.
Conditional independence of $\theta$ and $s$ given $\omega$ ensures that the agent's expected payoff conditional on the adviser's belief $s$ is likewise $s\cdot\mathbf  \w(\sigmah)$.

We first establish basic properties of the set $W$ and a formal equivalence between strategies in the original problem and mappings from messages to the set $W$. 

\begin{lemma}\label{lem:measurability}
The set $W$ is compact and convex. Every strategy $\sigma\sim(\sigmah(m))_{m\in \M}$ induces a measurable payoff mapping $g_\sigma:\M\to W$, given by $g_\sigma(m)=\w(\sigmah(m))$. Conversely, every measurable mapping $g:\M\to W$ can be implemented by a strategy $\sigma\sim(\sigmah(m))_{m\in \M}$ such that $\w(\sigmah(m))=g(m)$ for every $m\in\M$.
\end{lemma}
\begin{proof}
Let $\bar f=\sum_\omega\mu_0(\omega)f(\cdot\mid\omega)$ and $r_\omega=df(\cdot\mid\omega)/d\bar f$; the Radon--Nikodym derivatives exist because $\mu_0$ has full support. Every private strategy $\sigmah$ induces a joint distribution $\rho(d\theta,da)=\bar f(d\theta)\sigmah(da\mid\theta)$ in $\mathcal R\triangleq\{\rho\in\Delta(\Theta\times A):\rho_\Theta=\bar f\}$. Conversely, the disintegration theorem gives, for every $\rho\in\mathcal R$, a probability kernel $\sigmah(da\mid\theta)$ with $\rho(d\theta,da)=\bar f(d\theta)\sigmah(da\mid\theta)$. The payoff profile induced by $\rho$ is given by
\[
\Phi(\rho)(\omega)\triangleq\int_{\Theta\times A}r_\omega(\theta)u(a,\omega,\theta)\,\rho(d\theta,da).
\]
Thus, $W=\Phi(\mathcal R)$. The set $\mathcal R$ is compact and convex in the weak topology, and $\Phi$ is affine and continuous.\footnote{To establish continuity, we can first apply the fixed-marginal convergence result of \cite{JM, JMCorrection} to the integrand $r_\omega(\theta)u(a,\omega,\theta)$ truncated to $[-K,K]$. Since $r_\omega(\theta)\sup_a|u(a,\omega,\theta)|$ is integrable under $\bar f$ (by the uniform integrability assumption on $u$), and every $\rho\in\mathcal R$ has the same $\Theta$-marginal $\bar f$, the truncation error is bounded uniformly over $\mathcal R$ by $\int(r_\omega(\theta)\sup_a|u(a,\omega,\theta)|-K)_+\,d\bar f$, which vanishes as $K\to\infty$. Thus, the original integral is a uniform limit of continuous functions.} Thus $W$ is compact and convex.

For any strategy $\sigma$, the function $m\mapsto\w(\sigmah(m))(\omega)$ is measurable for each $\omega$, by integration of the measurable kernel $\sigma$ against the fixed conditional type distribution $f(\cdot\mid\omega)$. Conversely, for any measurable $g$, the measurable maximum theorem (Theorem~18.19 in \cite{AB}) applied to $-\|\Phi(\rho)-g(m)\|$ yields a measurable selection $m\mapsto\rho_m\in\mathcal R$ satisfying $\Phi(\rho_m)=g(m)$. The parameterized disintegration theorem (Theorem~9.27 in \cite{Kallenberg}), applied to the probability kernel $m\mapsto\rho_m$ with common $\Theta$-marginal $\bar f$, yields a probability kernel $\sigma(da\mid m,\theta)$ measurable in $(m,\theta)$ and satisfying $\rho_m(d\theta,da)=\bar f(d\theta)\sigma(da\mid m,\theta)$ for every $m$. Thus $\sigma$ implements $g$.
\end{proof}

\begin{remark}\label{remark_PF}
There is a close connection between Bayes-optimal private strategies and the Pareto frontier of the set $W$. A profile $\w\in W$ is induced by a Bayes-optimal private strategy at some belief if and only if it is \emph{weakly Pareto optimal} in
$W$: there is no $v\in W$ with $v(\omega)>\w(\omega)$ for every $\omega\in\Omega$. In one direction, a Bayes-optimal profile cannot be strictly improved in every state. For the converse, if $\w$ is weakly Pareto optimal, the
convex sets $W$ and $\w+\mathbb R_{++}^{N}$ are disjoint.
Since $W$ is compact and convex by \autoref{lem:measurability}, the separating-hyperplane theorem yields a nonzero vector
$p\in\mathbb R_+^{N}$ such that
$p\cdot\w\geq p\cdot v$ for all $v\in W$. Normalizing $p$ to sum to
one yields a belief at which $\w$ is Bayes-optimal.
\end{remark}

For the existence proof, it is helpful to further simplify the problem by observing that the agent's overall strategy determines a set of available payoff profiles (which we will refer to as a \textit{menu}) and a matching of messages to these profiles. For a fixed menu, the matching can be optimized as if the adviser were always aligned: each message is assigned a payoff profile that maximizes the agent's expected payoff at the corresponding belief. In contrast, when the adviser is misaligned, the worst-case payoff depends only on the menu, and is given by the minimum expected payoff over its elements. We formalize these intuitions next.

Let $\mathcal K_c(W)$ be the set of nonempty compact convex subsets of $W$. For $C\in\mathcal K_c(W)$, let $h_C(m)=\max_{\w\in C}m\cdot\w$ and $\ell_C(m)=\min_{\w\in C}m\cdot\w$, and define
\[
G(C)\triangleq\int_M\big[\alpha h_C(s)+(1-\alpha)\ell_C(s)\big]\tau(ds).
\]
For a measurable mapping $g:\M\to W$ and a reporting strategy $\beta$, we will abuse notation slightly by writing
\[
V(\beta,g)=\alpha\int_M s\cdot g(s)\,\tau(ds)
+(1-\alpha)\int_M\int_{\M} s\cdot g(m)\,\beta(dm\mid s)\tau(ds),
\]
noting that this definition agrees with the definition of $V(\beta,\sigma)$ whenever $g=g_\sigma$ (see \autoref{lem:measurability}). Similarly, we will write $V(g)=\inf_{\beta\in\mathcal B}V(\beta,g)$. Finally, let  $K_g=\overline{\operatorname{co}}\{g(m):m\in\M\}$, where $\overline{\operatorname{co}}$ denotes the closed convex hull operator.

\begin{lemma}\label{lem:menus}
For every measurable mapping $g:\M\to W$,
\[
V(g)=\int_M\big[\alpha s\cdot g(s)+(1-\alpha)\ell_{K_g}(s)\big]\tau(ds).
\]
Moreover, $V^*=\sup_{g:\M\to W}V(g)=\max_{C\in\mathcal K_c(W)}G(C)$, and if $C^*$ maximizes $G$, then any measurable selection $g^*(m)\in\arg\max_{\w\in C^*}m\cdot\w$ corresponds to an optimal strategy.
\end{lemma}
\begin{proof}
Since a linear functional has the same infimum over a set and its closed convex hull, $\inf_{m\in\M}s\cdot g(m)=\ell_{K_g}(s)$. Thus, $V(g)$ is bounded below by $\int_M\big[\alpha s\cdot g(s)+(1-\alpha)\ell_{K_g}(s)\big]\tau(ds)$. The reverse inequality also holds, because the misaligned adviser can approximate this bound using measurable reporting strategies.\footnote{Formally, we can choose messages $(m_j)_{j\geq1}$ such that $\{g(m_j):j\geq1\}$ is dense in the image of $g$. For each $k$, let $j_k(s)$ be the smallest index minimizing $s\cdot g(m_j)$ over $1\leq j\leq k$. Since this minimum is taken over finitely many continuous functions of $s$, the mapping $s\mapsto j_k(s)$ is measurable, and hence $\beta_k(\cdot\mid s)=\delta_{m_{j_k(s)}}$ defines a reporting strategy. As $k\to\infty$, the expected payoff at each belief $s$ converges to $\ell_{K_g}(s)$. Bounded convergence then implies that the misaligned adviser's expected payoff under $\beta_k$ converges to $\int_{\M} \ell_{K_g}(s)\tau(ds)$.}

The proven expression for $V(g)$ immediately implies that $V(g)\leq G(K_g)$. Conversely, for any menu $C$, the measurable maximum theorem gives a measurable selection $g_C(m)\in\arg\max_{\w\in C}m\cdot\w$, implemented by \autoref{lem:measurability}. When the adviser's true belief is $s$, this mapping yields $h_C(s)$ under truthful reporting and at least $\ell_C(s)$ under any reporting strategy of the misaligned adviser. Hence, $V(g_C)\geq G(C)$. Therefore $V^*=\sup_{C\in\mathcal K_c(W)}G(C)$.

The space $\mathcal K_c(W)$ is compact under Hausdorff distance: nonempty compact subsets of $W$ form a compact space (Theorem~3.85(3) in \cite{AB}), and convexity is preserved under Hausdorff limits. Both functions $h_C(m)$ and $\ell_C(m)$ change by at most $d_H(C,D)$ when $C$ is replaced by $D$, since $\|m\|\leq1$ on the space of beliefs. Hence, $G$ is continuous and has a maximizer $C^*$. By the equality we established above, any maximizing selection $g^*$ from $C^*$ guarantees $V^*$ and is therefore optimal.
\end{proof}

The final preliminary lemma proves that there exists a common minimax reporting strategy for all optimal strategies of the agent. This lemma is a key ingredient of the proof of the minimax result in \autoref{thm:minimax} but it is also used in the proof of \autoref{thm:trust}.

\begin{lemma}\label{lem:minimax-adviser}
There exists a reporting strategy $\beta^*\in\mathcal B$ such that
$V(\beta^*,\sigma)\leq V^*$ for every $\sigma\in\Sigma$.
For any optimal strategy $\sigma^*$, $(\beta^*,\sigma^*)$ is a saddle point:
\begin{equation}\label{eq:saddle}
V(\beta^*,\sigma)\leq V(\beta^*,\sigma^*)
\leq V(\beta,\sigma^*)
\qquad\text{for every }\sigma\in\Sigma,\ \beta\in\mathcal B.
\end{equation}
\end{lemma}
\begin{proof}
We first represent the adviser's reporting strategy by a joint distribution of the true belief and message in the compact convex set
\[
\Gamma=\{\chi\in\Delta(M\times M):\chi_1=\tau\},
\]
where $\chi_1$ and $\chi_2$ denote the marginal distributions of true beliefs and messages under $\chi$, respectively. 
Every $\chi\in\Gamma$ corresponds to a reporting strategy $\beta\in\mathcal B$ supported on $M$, since the disintegration theorem guarantees a probability kernel satisfying $\chi(ds,dm)=\tau(ds)\beta(dm\mid s)$. We will show that restricting attention to such reporting strategies is sufficient to construct a common minimax strategy.
Abusing notation again, for measurable $g:\M\to W$, we write
\[
V(\chi,g)
=\alpha\int_M s\cdot g(s)\,\tau(ds)
 +(1-\alpha)\int_{M\times M}s\cdot g(m)\,\chi(ds,dm),
\]
which agrees with $V(\beta,g)$ when $\chi$ is the joint distribution induced by $\beta$.

On the space $\Gamma\times\mathcal C(M,W)$, where $\mathcal C(M,W)$ is the convex set of continuous mappings from $M$ to $W$, the function $V(\chi,g)$ is affine in each argument, continuous in $\chi$ under the weak topology for fixed $g$, and continuous in $g$ under the uniform topology for fixed $\chi$. The minimax theorem of \cite{Sion1958Minimax} therefore yields
\[
v^*:=\sup_{g\in\mathcal C(M,W)}\min_{\chi\in\Gamma}V(\chi,g)
=\min_{\chi\in\Gamma}\sup_{g\in\mathcal C(M,W)}V(\chi,g).
\]
The minimum on the right-hand side is attained at some $\chi^*\in\Gamma$, since the supremum of continuous functions of $\chi$ is lower semicontinuous and $\Gamma$ is compact.

To extend the minimax inequality from continuous to all measurable mappings $g$, we now fix $\chi^*$ and use approximation in the $L^1(q^*)$ topology, where $q^*=\alpha\tau+(1-\alpha)\chi_2^*$ is the unconditional distribution of messages when the misaligned adviser uses $\chi^*$. Every measurable $g:M\to W$ is an $L^1(q^*)$ limit of continuous maps into $W$.\footnote{Approximate each coordinate by continuous functions (Theorems~12.7 and~13.9 in \cite{AB}), then use the fact that Euclidean projection onto the compact convex set $W$ is $1$-Lipschitz.} Since $\|s\|\leq1$ for every belief $s$, we have
\[
|V(\chi^*,g)-V(\chi^*,h)|
\leq\int_M\|g(m)-h(m)\|\,q^*(dm).
\]
Thus, the minimax inequality extends to every measurable payoff mapping $g$:
\begin{equation}\label{eq:upper}
V(\chi^*,g)\leq v^*
\qquad\text{for every measurable }g:M\to W.
\end{equation}
Let $\beta^*(dm\mid s)$ be a reporting strategy satisfying
$\chi^*(ds,dm)=\tau(ds)\beta^*(dm\mid s)$.
By \autoref{lem:measurability}, inequality \eqref{eq:upper} applies to the payoff mapping induced by every agent strategy, implying $V^*\leq v^*$. Conversely, for each continuous $g\in\mathcal C(M,W)$, let $C=\overline{\operatorname{co}}\{g(m):m\in M\}$. As in \autoref{lem:menus},
\[
\min_{\chi\in\Gamma}V(\chi,g)
=\int_M\left[\alpha s\cdot g(s)+(1-\alpha)\ell_C(s)\right]\tau(ds)
\leq G(C)\leq V^*.
\]
Thus, $v^*=V^*$, and $V(\beta^*,\sigma)\leq V^*$ for every agent strategy $\sigma$.
Finally, optimality of $\sigma^*$ gives $V(\beta,\sigma^*)\geq V^*$ for every $\beta$. Combining these two bounds establishes \eqref{eq:saddle}.
\end{proof}

\subsection{Completion of proof of \autoref{thm:trust}}\label{proof_thm:trust}

Existence of an optimal strategy follows from \autoref{lem:measurability}  and \autoref{lem:menus}.

Throughout the proof, we focus on the nontrivial case $\alpha>0$.\footnote{If $\alpha=0$, the agent can attain the optimal value $U(\mu_0)$ by ignoring the adviser, while the misaligned adviser can guarantee the corresponding upper bound by sending any fixed message. In particular, for an optimal strategy and any message $m$, sending $m$ with probability one gives the agent payoff $\mu_0\cdot g(m)$, so optimality requires $\mu_0\cdot g(m)=U(\mu_0)$. Every response is therefore Bayes-optimal at $\mu_0$. Fixing any message $m_0$, the constant strategy that uses $\sigmah^*(m_0)$ after every message is a TRS with $T=\{\mu_0\}$ and $P(m)=\mu_0$, equivalent to the original strategy under the adviser who always reports $m_0$.}
Fix an optimal strategy $\sigma^*$. The proof that $\sigma^*$ is equivalent to a TRS with a connected trust region proceeds in three steps.  First, we enlarge the menu of payoff profiles so that its maximizing profiles are Bayes-optimal at some belief, and construct an equivalent strategy that optimally maps messages to these profiles.
Second, we construct the corresponding trust region in the belief space. Finally, we show that responses outside the trust region can be implemented by projecting messages onto its visible boundary.

\medskip
\noindent\textit{Step 1: Enlarge the menu of payoff profiles.}
Let
$g(m)=\w(\sigmah^*(m))$ be the mapping from messages to state-contingent payoff profiles corresponding to the optimal strategy $\sigma^*$ (as in \autoref{lem:measurability}), $K_g=\overline{\operatorname{co}}\{g(m):m\in\M\}$, and
\[
C=(K_g+\mathbb R_+^{N})\cap W.
\]
The set $C$ is compact, convex, and upward closed within $W$. Adding payoff profiles that weakly improve on an existing profile in every state does not lower the worst payoff at any belief: $\ell_C=\ell_{K_g}$, while $h_C\geq h_{K_g}$ because $C\supset K_g$. Hence, \autoref{lem:menus} gives
\[
V^*=V(g)\leq G(K_g)\leq G(C)\leq V^*.
\]
The worst-case payoff against the misaligned adviser depends on the available payoff profiles, rather than their assignment to messages. Holding $C$ fixed, the agent can therefore maximize the expected payoff at each message as if the adviser were aligned (i.e., treating the message as the correct interim belief). Formally, the equality above gives
\[
0=G(C)-V(g)=\alpha\int_M
\big[h_C(s)-s\cdot g(s)\big]\,\tau(ds).
\]
Since $\al>0$ and the integrand is nonnegative, the set
$E=\{m\in\M:m\cdot g(m)=h_C(m)\}$ satisfies $\tau(E)=1$.
Choose a measurable mapping $\overline g:\M\to C$ that equals $g$ on $E$ and satisfies $\overline g(m)\in\arg\max_{\w\in C}m\cdot\w$ for every $m\in\M$. Such a mapping exists by the measurable maximum theorem. Retain the original private strategies $\sigmah^*(m)$ on $E$ and implement $\overline g$ elsewhere using \autoref{lem:measurability}. The resulting strategy $\overline\sigma$ is optimal by \autoref{lem:menus}.

Although $g$ and $\overline g$ coincide $\tau$-almost surely, they may induce different outcomes when the adviser is misaligned, since he can assign positive probability to messages in a $\tau$-null set. To establish equivalence of the original and modified agent strategies in terms of the induced outcome distributions, we construct a common minimax reporting strategy whose induced distribution of messages is absolutely continuous with respect to $\tau$. Not every optimal reporting strategy has this property. Intuitively, however, messages in a $\tau$-null set reveal that the adviser is misaligned, allowing the agent to respond to them independently of truthful communication. This allows us to modify a common minimax reporting strategy without compromising its minimax property.

Let $\beta^*$ be the common minimax strategy from
\autoref{lem:minimax-adviser}, and denote by
\[
\nu^*(A)=\int_M\beta^*(A\mid s)\tau(ds)
\]
the marginal distribution of messages under the joint distribution
$\tau(ds)\beta^*(dm\mid s)$, for every measurable set $A\subseteq M$.
By the Lebesgue decomposition theorem, there is a measurable set $Z\subseteq M$ such that $\tau(Z)=0$ and the restriction of $\nu^*$ to $M\setminus Z$ is absolutely continuous with respect to $\tau$. Define a reporting strategy $\beta^{**}:M\to\Delta(M)$ by
\[
\beta^{**}(A\mid s)=\beta^*(A\setminus Z\mid s)
+\beta^*(Z\mid s)\tau(A)
\]
for every Borel set $A\subseteq M$.
Thus, whenever the original strategy would send a message in $Z$, that message is replaced by a draw from $\tau$, independently of the true belief $s$. The marginal distribution of messages under
$\tau(ds)\beta^{**}(dm\mid s)$ is $\nu^{**}(A)=\nu^*(A\setminus Z)+\nu^*(Z)\tau(A)$,
so $\nu^{**}$ is absolutely continuous with respect to $\tau$.

Moreover, $\beta^{**}$ remains a common minimax reporting strategy. To see this, fix any measurable $h:\M\to W$, let $\bar w=\int_M h(t)\tau(dt)\in W$, and define $\overline h:\M\to W$ to equal $\bar w$ on $Z$ and $h$ elsewhere. For each $s\in M$,
\[
\int_M s\cdot h(m)\,\beta^{**}(dm\mid s)
=\int_{M\setminus Z}s\cdot h(m)\,\beta^*(dm\mid s)
   +\beta^*(Z\mid s)\,s\cdot\bar w 
=\int_M s\cdot \overline h(m)\,\beta^*(dm\mid s).
\]
The aligned-adviser contributions also agree because $\tau(Z)=0$.
Consequently,
\[
V(\beta^{**},h)=V(\beta^*,\overline h)\leq V^*.
\]
The last inequality follows from \autoref{lem:minimax-adviser}, since $\overline h$ is implementable by \autoref{lem:measurability}. Thus, $\beta^{**}$ remains a common minimax reporting strategy and is adversarial to every optimal agent strategy, including $\sigma^*$ and $\overline\sigma$.

Accounting for both adviser types, the unconditional distribution of messages under $\beta^{**}$ is
\[
q^*=\alpha\tau+(1-\alpha)\nu^{**}\ll\tau.
\]
Since $\sigma^*$ and $\overline\sigma$ use the same private strategies on $E$ and $q^*(E)=1$, they induce the same joint distribution of states, types, messages, and actions against $\beta^{**}$. Thus, $\sigma^*$ and $\overline\sigma$ are equivalent.

\medskip
\noindent\textit{Step 2: Construct a compact, connected trust region.} Fixing the equivalent optimal strategy $\overline\sigma$ and the (compact and convex) menu $C$ constructed in step 1, define $F(m)=\arg\max_{\w\in C}m\cdot\w$, and let $D=\bigcup_{m\in\M}F(m)\subseteq C$ be its image. Intuitively, $D$ is the subset of payoff profiles of $C$ that are an optimal response for the agent for some message $m$ (treating that message as the correct interim belief). The correspondence $F$ has a compact graph and nonempty convex values over the connected domain $\M$. Consequently, its graph is connected, and so is its image $D$.\footnote{A compact graph with connected values over a connected domain is connected: otherwise, it could be partitioned into two disjoint nonempty closed sets whose projections would likewise partition the domain into two disjoint nonempty closed sets.}
Since $C$ is upward closed, every $\w\in D$ is weakly Pareto optimal in $W$. Indeed, if $v\in W$ strictly improved on $\w$ in every state, then $v\in C$; but at any belief $m$ for which $\w\in\arg\max_{w'\in C}m\cdot w'$, we would have $m\cdot v>m\cdot\w$, a contradiction.

Since every payoff profile in $D$ is weakly Pareto optimal in $W$, it is Bayes-optimal at some belief, by \autoref{remark_PF}. For each $\w\in D$, let
\[
B(\w)=\{\mu\in\M:\mu\cdot\w=U(\mu)\}
\]
denote the set of beliefs at which $\w$ is Bayes-optimal. This set is nonempty, compact, and convex. Moreover, the correspondence $B$ has a compact graph over the connected set $D$. By the same argument as above, its graph and image are connected. We can therefore define the nonempty, compact, and connected trust region
\[
T=\bigcup_{\w\in D}B(\w)
=\{\mu\in\M:h_C(\mu)=U(\mu)\}.
\]
By construction, $h_C(\mu)=U(\mu)$ for every $\mu\in T$, so any payoff profile maximizing $\mu\cdot\w$ over $C$ is also Bayes-optimal at $\mu$. Thus, $\overline\sigma$ assigns a Bayes-optimal private strategy to every message $m\in T$, satisfying the first requirement of a TRS.

\medskip
\noindent\textit{Step 3: Map untrusted messages to the visible boundary.}
For every $m\notin T$, choose a measurable mapping $\overline\mu$ satisfying $\overline\mu(m)\in B(\overline g(m))$. Such a mapping exists by the measurable maximum theorem applied to the function $(m,\mu)\mapsto \mu\cdot\overline g(m)-U(\mu)$ on $\M\times\M$. The payoff profile $\overline g(m)$ maximizes expected payoff over $C$ at $m$ by construction, and at $\overline\mu(m)$ because it is Bayes-optimal there. By linearity of expected payoffs in beliefs, the same profile maximizes expected payoff over $C$ at every belief on the segment joining $m$ and $\overline\mu(m)$.

For every $m\notin T$, let $P(m)$ be the first point of $T$ encountered when moving along the segment from $m$ toward $\overline\mu(m)$, and set $P(m)=m$ for $m\in T$. This projection is well-defined because $T$ is compact and $\overline\mu(m)\in T$.\footnote{To verify measurability on $\M\setminus T$, let $\lambda_*(m)$ be the smallest $\lambda\in[0,1]$ such that $(1-\lambda)m+\lambda\overline\mu(m)\in T$. For any $r\in[0,1]$, the condition $\lambda_*(m)\leq r$ is equivalent to $\min_{\lambda\in[0,r]}\operatorname{dist}((1-\lambda)m+\lambda\overline\mu(m),T)=0$. This minimum is measurable in $m$, since it equals the infimum over a countable dense subset of $[0,r]$ of measurable functions. Hence $\lambda_*$ and $P$ are measurable.}
By construction, $P(m)$ is visible from $m$. Moreover, the preceding argument shows that $\overline g(m)$ maximizes expected payoff over $C$ at $P(m)$. Since $P(m)\in T$, the private strategy prescribed by $\overline\sigma$ after message $m$ is Bayes-optimal at $P(m)$. Thus, $\overline\sigma$ is a TRS.

\paragraph{Proof of \autoref{cor:non-hollow}.}\label{app:non-hollow}
The singleton trust region for $\al=0$ is non-hollow. For $\al>0$, take $\mu\in\M\setminus T$ and a Bayes-optimal profile $\w$ at $\mu$. Then $\w\notin C$. Since $\w\in W$ and $C$ is upward closed within $W$, we also have $\w\notin C+\mathbb R_+^N$. Separating $\w$ from the closed convex set $C+\mathbb R_+^N$ yields a nonzero nonnegative vector, which can be normalized to a belief $s\in\M$ satisfying $s\cdot\w<\ell_C(s)$.

Consider any $q=\mu+\lambda(\mu-s)\in\M$, $\lambda>0$, and any Bayes-optimal profile $v$ at $q$. Optimality gives
\[
\mu\cdot(v-\w)\leq0\leq q\cdot(v-\w)
=(1+\lambda)\mu\cdot(v-\w)-\lambda s\cdot(v-\w).
\]
Thus $s\cdot v\leq s\cdot\w<\ell_C(s)$, so no Bayes-optimal profile at $q$ belongs to $C$, implying $q\notin T$. Therefore, starting from any $\mu\notin T$, we can move along the ray pointing away from $s$ without encountering $T$, establishing non-hollowness.

\paragraph{Proof of \autoref{cor:projection}.}\label{app:projection}

Fix any optimal strategy, and let $\sigma$ be the equivalent optimal TRS constructed in the proof of \autoref{thm:trust}.  Suppose first that $\alpha>0$.
By construction, $\sigmah(m)$ maximizes the agent's
payoff at $m$ over the menu $C$, which contains
the payoff profiles of all private strategies in
$\operatorname{Im}(\sigma)$. Moreover, for every
$m\notin T$, $\sigmah(m)$ is Bayes-optimal at $P(m)$.
Thus,
\[
d_\sigma(m,P(m))
=U(m)-U(\sigmah(m),m)
\leq d_\sigma(m,\mu)
\qquad\text{for every }\mu\in T.
\]
For $\alpha=0$, the result follows immediately from
the singleton trust region construction in the proof
of \autoref{thm:trust}.

\subsection{Completion of the proof of \autoref{thm:minimax}}\label{proof_thm:minimax}

By \autoref{lem:menus} and \autoref{thm:trust}, there exists an optimal TRS $\sigma^*$ with trust region $T$ and measurable map $P$. Pairing it with the adversarial reporting strategy $\beta^*$ from \autoref{lem:minimax-adviser}, we obtain a saddle point $(\beta^*,\sigma^*)$. It remains to establish that $\sigma^*$ prescribes Bayes-optimal responses to the posterior beliefs induced by $\beta^*$.

Let
\[
q^*(A)=\alpha\tau(A)+(1-\alpha)\int_M\beta^*(A\mid s)\tau(ds),
\qquad A\subseteq M,
\]
be the resulting distribution of messages. For each $\omega\in\Omega$, define the measure
\[
\nu_\omega(A)=\alpha\int_A s(\omega)\tau(ds)
 +(1-\alpha)\int_M s(\omega)\beta^*(A\mid s)\tau(ds),
\qquad A\subseteq M.
\]
Since $\sum_{\omega\in\Omega}\nu_\omega=q^*$, the Radon--Nikodym derivatives
$p_\omega=d\nu_\omega/dq^*$ exist and define a measurable posterior
$p(m)=(p_\omega(m))_{\omega\in\Omega}\in\M$ for $q^*$-almost every $m$.

By inequality \eqref{eq:saddle}, $\sigma^*$ is a best response to $\beta^*$. Since every measurable payoff mapping is implementable by \autoref{lem:measurability}, the agent's prescribed response must maximize expected payoff at the posterior belief $p(m)$ for $q^*$-almost every message $m$. Thus,
\[
p(m)\cdot\w(\sigmah^*(m))=U(p(m))
\qquad q^*\text{-almost surely}.
\]
Let $D\subseteq M$ be the measurable set on which this equality holds. Then $q^*(D)=1$.

At each remaining message $m\notin D$, define the posterior belief to be $m$ if $m\in T$ and $P(m)$ otherwise. These beliefs rationalize the prescribed responses because $\sigma^*$ is a TRS. Since $q^*(D)=1$, this extension of the posterior mapping leaves the joint distribution of states and messages unchanged. Thus $\sigma^*$ is robustly rationalizable.

\subsection{Proof of \autoref{thm:viable}}\label{proof_thm:viable}
We first establish the bounds and attainment result for finite $M$, with $|M|=K$. We then extend the bounds to arbitrary supports.

\textbf{Notation.} In this appendix, we denote by $I_K$ a unit matrix of dimension $K$, by $1_K$ a vector of ones of dimension $K$, by $0_{N\times K}$ a matrix of zeros of dimension $N\times K$, by $e^i_K$ an $i$th standard basis vector of dimension $K$, by $x_{i,k}$ a $k$th element of vector $x_i$, by $\mathrm{diag}\,x$ a diagonal matrix with vector $x$ on the main diagonal, by $\delta_{ij}$ the Kronecker delta (equal to $1$ if $i=j$ and $0$ otherwise), and by $X^\top$ a transpose of a matrix $X$.

Since $\mu_0$ has full support, we can equivalently identify the adviser's information with a (row) stochastic $N\times K$ matrix $\Tech$, where $\Tech_{ij}$ is the probability of the $j$th signal observed by the adviser in the $i$th state. Moreover, $\rank \Tech=R(\tau)$.\footnote{A matrix of adviser's posteriors can be computed by Bayes' rule as $(\mu(s))_{s\in S}=(\mathrm{diag}(\mu_0(\omega))_{\omega\in\Omega})\Tech (\mathrm{diag}(\tau(s))_{s\in S})^{-1}$. The diagonal matrices are invertible and the multiplication by an invertible matrix preserves the rank.}

We identify the  strategy of the misaligned adviser with a stochastic $K\times K$ matrix $\Sen$. Since the aligned adviser reports truthfully, the overall adviser's strategy can be written as a garbling of his information:
\begin{align}
    G(\Sen)\triangleq \al I_K+(1-\al)\Sen.
\end{align}
By \cite{Blackwell1951Comparison}, the $\MVA$ is a maximal $\al$ for which there exists a stochastic matrix $\Sen$ such that $\Tech G(\Sen)$ is Blackwell uninformative; we show below that it is attained. This also implies that $\MVA$ depends on $\tau$ only via $\Tech$, so we will write $\MVA(\Tech)$. 

We start with preliminary observations. First, note that $G_{kk}\geq\al$ and $G 1_K=1_K$. Second, note that $\Tech G(\Sen)$ is uninformative if and only if all of its rows are equal to each other, that is if and only if  
\begin{align}\label{eq:D_1}
    D(\Tech)G(\Sen)=D(\Tech)(\al I_K+(1-\al)\Sen)=0_{(N-1)\times K},
\end{align}
where $D(\Tech)$ is the row-difference matrix of $\Tech$:
\begin{align*}
    D(\Tech)\triangleq \begin{pmatrix}(\tech_2-\tech_1)^\top\\ \vdots\\ (\tech_N-\tech_1)^\top\end{pmatrix},
\end{align*}
and $\tech_i^\top$ is the $i$th row of $\Tech$. Consider the auxiliary finite linear program:
\begin{align}\label{eq:alpha-lp}
   \Lambda(\Tech)\triangleq\max_{G\in\reals^{K\times K},\,\al\in\reals}\quad & \al\\
\text{s.t.}\quad
& G\geq \al I_{K},\ G 1_K=1_K,\label{eq:alpha-lp-1}\\ 
& D(\Tech) G = 0_{(N-1)\times K}\label{eq:alpha-lp-2}.
\end{align}
The first lemma identifies the program's value with $\MVA$; the second establishes the rank bounds and the linearly independent case; the third constructs information structures attaining every value in $[1/N,1/2]$.
\begin{lemma}\label{lem:alpha-1}
$\MVA(\Tech)=\Lambda(\Tech)$.
\end{lemma}
\begin{proof}
We need to show that there exists a stochastic matrix $\Sen$ such that $\Tech G(\Sen)$ is Blackwell uninformative if and only if $\al\leq \Lambda(\Tech)$.

\textbf{Only if:} For any given $\al$, if $\Sen$ is such that $\Tech G(\Sen)$ is Blackwell uninformative, then we showed that $G(\Sen)$ must satisfy  conditions  (\ref{eq:alpha-lp-1}-\ref{eq:alpha-lp-2}). By the maximization nature of the problem, if $\al>\Lambda(\Tech)$, those conditions cannot be satisfied.

\textbf{If:} If $\al\leq\Lambda(\Tech)$, then there exists $G$ that satisfies  conditions (\ref{eq:alpha-lp-1}-\ref{eq:alpha-lp-2}) (e.g., the argmax). If $\Lambda(\Tech)=1$, then $\Sen$ can be arbitrary. Otherwise, set $\Sen=(G-\al I_K)/(1-\al)$. It is straightforward that the resulting $\Sen$ is a stochastic matrix and by construction  $\Tech G(\Sen)$ is uninformative.
\end{proof}
\begin{lemma}
$\MVA(\Tech)\in[1/R(\Tech),1/2]$. If $R(\Tech)=K$, then $\MVA(\Tech)=1/K$. 
\end{lemma}
\begin{proof}To ease notation, in the proof we omit the dependence of $R$ on $\Tech$.

\noindent\textbf{1.)} $\MVA(\Tech)\leq 1/2$. 

If $\al$ and $G$ satisfy (\ref{eq:alpha-lp-1}-\ref{eq:alpha-lp-2}), then $\Sen=(G-\al I_K)/(1-\al)$ is a stochastic matrix and 
\begin{align}\label{eq:eigen}
    D(\Tech)\Sen=-\frac{\al}{1-\al}D(\Tech).
\end{align}
In other words, the rows of $D(\Tech)$ are left eigenvectors of $\Sen$ associated with eigenvalue $-\al/(1-\al)$. Since $\Sen$ is stochastic, its spectral radius equals $1$. Thus, $|-\al/(1-\al)|\leq 1$ and $\al\leq 1/2$. It follows that $\MVA(\Tech)\leq 1/2$.

\noindent\textbf{2.)} If $R=K$, then $\MVA(\Tech)=1/K$. 

If $R=K$, then $K\leq N$ and $\mathrm{rank} D(\Tech)=K-1$. Thus, $\mathrm{rank}\,\mathrm{ker} D(\Tech)=K-(K-1)=1$ and, because $D1_K=1_{N-1}-1_{N-1}=0_{N-1}$, $\mathrm{ker}D=\mathrm{span}\{1_K\}$. Thus, for $(G,\al)$ to satisfy (\ref{eq:alpha-lp-2}),  every column of $G$ must be a multiple of $1_K$. But since $G$ is stochastic, it follows that $\sum_{k=1}^K G_{kk}=1$ and $\min_k G_{kk}\leq 1/K$. To further satisfy (\ref{eq:alpha-lp-1}), it must be that  $\al\leq 1/K$. Thus, $\MVA(\Tech)\leq 1/K$.

At the same time, if $\al\leq 1/K$, then $(G,\al)$ satisfy  (\ref{eq:alpha-lp-1}-\ref{eq:alpha-lp-2}) for $G=1/K 1_K 1^\top_K$.  In this case, $G$ is uninformative, not only $\Tech G$,  so the misaligned adviser can make the signal uninformative about his estimate, not only about the state.  It follows that $\MVA\geq 1/K$ and, therefore, $\MVA(\Tech)=1/K$.

\noindent\textbf{3.)} $\MVA\geq 1/R$. 

Let  $\al=1/R$ (recall that $R\geq 2$). Consider the normed space $(\reals^K,\|\cdot\|_1)$ and its linear $(R-1)$-dimensional subspace $\mathbb W$ spanned by rows of $D(\Tech)$. By the Auerbach basis theorem, there exist vectors $w_1,\dots,w_{R-1}\in \mathbb W$ and $x_1,\dots,x_{R-1}\in \reals^K$ such that\footnote{By the Auerbach theorem, there exist $w_1,\dots,w_{R-1}\in \mathbb W$ and $\phi_1,\dots,\phi_{R-1}\in \mathbb W^*$ such that $\|w_i\|_1=1$, $\|\phi_i\|_{\mathbb W^*}=1$, and $\phi_i(w_j)=\delta_{ij}$ (Section II.E, Lemma 11 in \cite{Wojtaszczyk1991Banach}; see also \cite{GershkovMoldovanuShi2025OrderIndependence} for another recent application). By the Hahn--Banach theorem, these $\phi_i$, operating on $\mathbb W$, can be extended to $\tilde\phi_i$, operating on $\reals^K$, without a change in their norm. By the duality between spaces $l_1$ and $l_\infty$, for each $i$, there exists $x_i\in\reals^K$ such that $\tilde\phi_i(z)=z^\top x_i$ and $\|x_i\|_\infty=\|\tilde\phi_i\|=1$. Then $w_i^\top x_j=\tilde\phi_j(w_i)=\phi_j(w_i)=\delta_{ij}$ for all $1\leq i,j\leq R-1$.}
\begin{align*}
\|w_i\|_1=1,\quad \|x_i\|_\infty=1,\quad w_i^\top x_j=\delta_{ij},\quad 1\leq i,j\leq R-1.
\end{align*}
Define the corresponding matrices $W\triangleq (w_1,\dots,w_{R-1})$, $X\triangleq (x_1,\dots,x_{R-1})$. By construction,
\begin{align}\label{eq:auerbach-1}
    W^\top X=I_{R-1},
\end{align}
and by properties of $D(\Tech)$,
\begin{align}\label{eq:auerbach-2}
    W^\top 1_K=0_{R-1}.
\end{align}
Define the vector of weights of rows of $W$, $\overline w\in \reals^K$, as $\overline w_k\triangleq \sum_{i=1}^{R-1}|w_{i,k}|$. Since $\|w_i\|_1=1$, we have
\begin{align}\label{eq:auerbach-3}
    \sum_{k=1}^K \overline w_k=\sum_{i=1}^{R-1}\|w_i\|_1=R-1.
\end{align}
We explicitly construct the desired strategy of the misaligned adviser $\Sen$ as:
\begin{align}
    \Sen=\frac{1}{R-1}(1_K \overline w^\top-XW^\top).
\end{align}

\noindent\emph{Nonnegativity.} For all $j,k$,
\begin{align*}
    \Sen_{jk}=\frac{1}{R-1}(\sum_{i=1}^{R-1}|w_{i,k}|-\sum_{i=1}^{R-1} x_{i,j} w_{i,k})\geq0,
\end{align*}
because $|x_{i,j}|\leq \|x_i\|_\infty=1$.

\noindent\emph{Stochasticity.} By (\ref{eq:auerbach-2}) and (\ref{eq:auerbach-3}):
\begin{align*}
    \Sen 1_K=\frac{1}{R-1}(1_K (\overline w^\top 1_K)-X(W^\top 1_K))=\frac{1}{R-1}(1_K (R-1)-0_{K})=1_K.
\end{align*}

\noindent\emph{Uninformativeness.} By (\ref{eq:auerbach-1}) and  (\ref{eq:auerbach-2}):
\begin{align*}
    W^\top \Sen=\frac{1}{R-1}((W^\top 1_K) \overline w^\top-(W^\top X)W^\top)=\frac{1}{R-1}(0_{R-1}-W^\top)=-\frac{1}{R-1}W^\top.
\end{align*}
Since by construction columns of $W$ form a basis in the row space of $D(\Tech)$, it follows that
\begin{align*}
    D(\Tech)\Sen=-\frac{1}{R-1}D(\Tech).
\end{align*}
As $\al=1/R$, this corresponds exactly to constraint \eqref{eq:alpha-lp-2}. The result follows.
\end{proof}

\begin{lemma} For any $N\geq2$ and $\al\in[1/N,1/2]$, there exist $K$ and $\Tech$ such that $\MVA(\Tech)=\al$.
\end{lemma}
\begin{proof}
The proof is by direct construction. For $N=2$, the result is trivial. For $N\geq3$, consider $K\in[4,N+1]$ and, for $\delta\in[0,1]$, the $N\times K$ matrix $\Tech$ such that 
\begin{align*}
\tech_i^\top &= \frac{1}{K} 1_K^\top,\quad i=1\ \textrm{or}\  i=K,K+1,\dots,N,\\
\tech_i^\top &= \frac{1}{K} (1_K+e^i_K-e^1_K)^\top, \quad i=2,\dots,K-2,\\
\tech_{i}^\top &=\frac{1}{K}( 1_K+e^{i}_K- \delta e^1_K-(1-\delta)e^K_K)^\top, \quad i=K-1,
\end{align*}
where $e^i_K$ is the $i$th basis vector of $\reals^K$.
By construction, $\Tech$ is a stochastic matrix. Consider $\MVA(\Tech)$ that solves the corresponding problem (\ref{eq:alpha-lp}).

The constraint $D(\Tech) G = 0_{(N-1)\times K}$ reduces to:
\begin{align}\label{eq:MVA-construction}
(e^i_K-e^1_K)^\top G=0,\quad i=2,\dots,K-2,
\qquad
(e^{K-1}_K-\delta e^1_K-(1-\delta)e^K_K)^\top G=0, 
\end{align}
which effectively states that the first $K-2$ rows are equal to each other and the $(K-1)$th row is a convex combination of the $1$st and the $K$th rows with weight $\delta$. Thus, the effective variables are the $1$st and the $K$th rows of the matrix $G$. The constraints $G\geq \al I_{K\times K}$ and $G 1_K=1_K$ then reduce to those rows being probability vectors, such that
\begin{gather*}
G_{1k}\geq \al,\ k=1,\dots,K-2,\ (\delta G_{1,K-1}+(1-\delta)G_{K,K-1})\geq \al,\ G_{KK}\geq \al.
\end{gather*}
Therefore,
\begin{align*}
\al \leq (\delta G_{1,K-1}+(1-\delta)G_{K,K-1})
\leq \delta(1-(K-2)\al)+(1-\delta)(1-\al)
=1-\al\big(1+\delta(K-3)\big).
\end{align*}
Rearranging yields
\begin{align}
\al \leq \al^\dagger\triangleq\frac{1}{2+\delta(K-3)},
\end{align}
and thus $\MVA(\Tech)\leq \al^\dagger$. Whenever $\delta\geq (K-4)/(K-3)$, $\al^\dagger\leq 1/(K-2)$ and the bound $\al^\dagger$ can be attained by $G$ with the $1$st and the $K$th rows being (the rest of $G$ is pinned down by condition \eqref{eq:MVA-construction}):
\begin{align*}
    G_{1k}&=\al^\dagger,\ k=1,\dots,K-2,&\ G_{1,K-1}&=1-(K-2)\al^\dagger,&\ G_{1K}&=0,\\
    G_{Kk}&=0,\ k=1,\dots,K-2,&\ G_{K,K-1}&=1-\al^\dagger,&\ G_{KK}&=\al^\dagger.
\end{align*}
Thus, $\MVA(\Tech)=\al^\dagger$. At $\delta=(K-4)/(K-3)$, $\al^\dagger=1/(K-2)$; at $\delta=1$, $\al^\dagger=1/(K-1)$.

This establishes that, for all $K\in[4,N+1]$, as $\delta$ spans $[(K-4)/(K-3),1]$, the proposed $\Tech$ achieves $\MVA(\Tech)$ that spans $[1/(K-1),1/(K-2)]$. Spanning $K$ from $4$ to $N+1$, we obtain the result.
\end{proof}

\paragraph{Infinite $\boldsymbol{M}$.}
We now extend the bounds to arbitrary supports of the adviser's beliefs. Recall the compact set $\Gamma$ of joint distributions of true beliefs and messages from \autoref{lem:minimax-adviser}. For any $\chi\in\Gamma$, the overall message is uninformative about the state if and only if the agent's posterior belief equals the prior $\mu_0$ almost surely. Equivalently, for every continuous function $\varphi:M\to\mathbb R$,
\begin{align*}
L_{\al,\chi}(\varphi)\triangleq
\al\int_M(s-\mu_0)\varphi(s)\tau(ds)
+(1-\al)\int_{M\times M}(s-\mu_0)\varphi(m)\chi(ds,dm)=0.
\end{align*}
For each $\varphi$, this condition defines a closed subset of $[0,1]\times\Gamma$, since $L_{\al,\chi}(\varphi)$ is continuous in $(\al,\chi)$. By compactness, the largest $\al$ for which all these conditions hold is attained. Moreover, if the overall message can be made uninformative at some $\al$, the same is true at any smaller alignment probability.

To establish the upper bound, choose two states $\omega,\omega'$ with distinct conditional distributions of the adviser's belief, and let $\eta=\pi(\cdot\mid\omega)-\pi(\cdot\mid\omega')$. Uninformativeness requires
$\alpha\eta+(1-\alpha)\eta\beta=0$,
where $\eta\beta$ denotes the signed measure defined by
$(\eta\beta)(A)=\int_M\beta(A\mid s)\eta(ds)$.
Since applying the same reporting kernel to two distributions cannot increase their total variation distance, we have
$\|\eta\beta\|_{\mathrm{TV}}\leq\|\eta\|_{\mathrm{TV}}$.
As $\eta\neq0$, it follows that $\alpha\leq1/2$.

To establish the lower bound, we approximate the adviser's information by finite information structures, to which the previously established bound applies. Partition $M$ into finitely many measurable sets of diameter at most $\epsilon$, and replace each belief $s$ by the conditional mean of its set. Denote this mean by $\bar s_\epsilon(s)$. The resulting finite information structure has rank at most $R(\tau)$, so the finite-case result guarantees a reporting strategy that makes the overall message uninformative at $\al=1/R(\tau)$. We can now use this strategy to construct a reporting strategy for the original information structure. Conditional on the adviser's true belief $s$, first apply the finite reporting strategy to the set containing $s$, obtaining another set in the partition. Then draw the actual message from $\tau$ conditional on the selected set. Let $\chi_\epsilon\in\Gamma$ be the resulting joint distribution of true beliefs and messages. Since the finite reporting strategy makes the selected set uninformative about the state, the only remaining difference comes from truthful reporting of the original belief $s$ rather than its conditional mean $\bar s_\epsilon(s)$. Specifically, for every continuous $\varphi:M\to\mathbb R$,
\begin{align*}
L_{\al,\chi_\epsilon}(\varphi)
&=\al\int_M(s-\bar s_\epsilon(s))\varphi(s)\tau(ds),\\
\|L_{\al,\chi_\epsilon}(\varphi)\|
&\leq\al\epsilon\|\varphi\|_\infty.
\end{align*}
By compactness of $\Gamma$, there is a subsequence of $\chi_\epsilon$ converging weakly as $\epsilon\to0$. The bound above implies that its limit satisfies $L_{\al,\chi}(\varphi)=0$ for every continuous~$\varphi$. Thus, messages can be made uninformative at $\al=1/R(\tau)$, establishing the lower bound.

Finally, note that if the misaligned adviser can make the overall message uninformative, advice has no value in any decision problem, so $\Delta V=0$. Conversely, when $\al>\MVA(\tau)$, the common minimax reporting strategy $\beta^*$ from \autoref{lem:minimax-adviser} necessarily induces a nondegenerate distribution of posterior beliefs. In a regular decision problem, such information has strictly positive value: a prior-optimal private strategy cannot also be optimal at a distinct posterior, since it would then be optimal at two distinct interior beliefs along the segment joining them, contradicting regularity. By the saddle-point property, the agent's optimal payoff against $\beta^*$ equals $V^*$, so $\Delta V>0$.

\subsection{Proof of \autoref{lem:binary-misaligned-adviser}}\label{proof_lem:binary-misaligned}

By definition of a TRS, the misaligned adviser with signal realization $\mu$ minimizes $U(\sigmah(\mu'),\mu)$ over $\mu'$ in the trust region. Recall that the function $U(\sigmah(\mu'),\mu)$ is linear in $\mu$, and we assumed that $U(\mu)=\max_{\sigmah}U(\sigmah,\mu)$ is strictly convex and twice differentiable in $\mu$. This means that $U(\sigmah(\mu'),\mu)$ is the value at $\mu$ of the hyperplane supporting $U$ at $\mu'$. Under our convention that $\mu$ is the probability of state $1$, this means that 
\begin{align*}
U(\sigmah(\mu'),\mu)=U(\mu')+U'(\mu')(\mu-\mu').
\end{align*}
By convexity of $U$, this function is quasi-concave in $\mu'$, and hence 
for all $\mu'\in[\muu,\muo]$, $U(\sigmah(\mu'),\mu)\geq \min \{U(\sigmah(\muu),\mu),U(\sigmah(\muo),\mu)\}$. Thus, the misaligned adviser's strategy takes a threshold form. The threshold $\muh(\muu,\muo)$ is the intersection point of the supporting lines to $U$ at points $\muu$ and $\muo$:
\begin{align*}
    U(\muu)+U'(\muu)(\muh(\muu,\muo)-\muu)=U(\muo)+U'(\muo)(\muh(\muu,\muo)-\muo).
\end{align*}
Rearranging, we obtain:
\begin{align*}
\muh(\muu,\muo)=\frac{\muo U'(\muo)-\muu U'(\muu)-(U(\muo)-U(\muu))}{U'(\muo)-U'(\muu)}.
\end{align*}
Applying integration by parts, the numerator equals $\int_{\muu}^{\muo} x U''(x)dx$ and the denominator equals $\int_{\muu}^{\muo} U''(x)dx$. The result follows.

\subsection{Proof of \autoref{prop:binary-trust-region}}\label{app:proof-binary-trust}

At $\al=1$, full trust is optimal, so $T=[0,1]$. When $\al=0$, ignoring the adviser is optimal and is rationalized by the constant reporting strategy $\beta(\mu)=\mu_0$. For the rest of the proof, assume $\al\in(0,1)$. By \autoref{lem:binary-misaligned-adviser} and \autoref{cor:binary-interval}, the choice of an optimal strategy for the agent reduces to optimization over the extreme points $\muu$, $\muo$ of the trust interval with the corresponding payoff:
\begin{align*}
    V&(\muu,\muo)\triangleq\\
    &\al\left(\int_0^{\muu} (U(\muu)+U'(\muu)(\mu-\muu)) \tau(\mu)d\mu +\int_{\muu}^{\muo}U(\mu)\tau(\mu)d\mu+\int_{\muo}^1 (U(\muo)+U'(\muo)(\mu-\muo))\tau(\mu)d\mu \right)\\
    &+(1-\al)\left( \int_0^{\muh(\muu,\muo)}(U(\muo)+U'(\muo)(\mu-\muo))\tau(\mu)d\mu +\int_{\muh(\muu,\muo)}^1  (U(\muu)+U'(\muu)(\mu-\muu))\tau(\mu)d\mu \right).
\end{align*}

The function $V(\muu,\muo)$ is continuously differentiable with  partial derivatives (whenever $\muu<\muo$):
\begin{align*}
    \frac{\partial V}{\partial \muu}&=U''(\muu)\left(\al\int_0^{\muu}(\mu-\muu) \tau(\mu)d\mu +(1-\al)\int_{\muh(\muu,\muo)}^1 (\mu-\muu) \tau(\mu)d\mu \right),\\
    \frac{\partial V}{\partial \muo}&=U''(\muo)\left(\al \int_{\muo}^1  (\mu-\muo)\tau(\mu)d\mu+(1-\al) \int_{0}^{\muh(\muu,\muo)} (\mu-\muo)\tau(\mu)d\mu\right).
\end{align*}
Intuitively, the first-order impact of a change in the trust boundary equals the change in the action played at that boundary, measured by $U''(\cdot)$, integrated over the belief regions in which the aligned and misaligned advisers induce that action, weighted by the alignment parameter. (Terms involving $\partial\muh/\partial\muu$ and $\partial\muh/\partial\muo$ vanish because at $\mu=\muh(\muu,\muo)$ the misaligned adviser is indifferent between the two messages.)

Whenever the trust region is non-singleton, $\muu<\muo$, at the optimal choice of $\muu$ and $\muo$ these partial derivatives must equal zero, $\partial V/\partial\muu=0$ and $\partial V/\partial\muo=0$. Since $U''(\cdot)>0$, these first-order conditions can be rearranged as follows. Define functions $\Psi_1$ and $\Psi_2$ as
\begin{align*}
      \Psi_1(\muu,\muo,\al)&\triangleq\al\int_0^{\muu}(\mu-\muu) \tau(\mu)d\mu +(1-\al)\int_{\muh(\muu,\muo)}^1 (\mu-\muu) \tau(\mu)d\mu,\\
      \Psi_2(\muu,\muo,\al)&\triangleq
    \al\int_{\muo}^{1}(\mu-\muo) \tau(\mu)d\mu +(1-\al)\int_0^{\muh(\muu,\muo)} (\mu-\muo) \tau(\mu)d\mu.
\end{align*}
Then, $ \Psi_1(\muu,\muo,\al)= \Psi_2(\muu,\muo,\al)=0$ is equivalent to  conditions (\ref{eq:cond_1}) and (\ref{eq:cond_2}).

First, we show that conditions (\ref{eq:cond_1}) and (\ref{eq:cond_2}) are incompatible with $\al<1/2$. Indeed, if those conditions hold then (for the rest of the proof, we will often omit the arguments of the function $\muh$ for brevity):
\begin{align*} 
    \al\left(\int_0^{\muh} (\muh-\mu)\tau(\mu)d\mu+\int_{\muh}^1 (\mu-\muh)\tau(\mu)d\mu\right)
    &\geq \al\left(\int_0^{\muu} (\muu-\mu)\tau(\mu)d\mu+\int_{\muo}^1 (\mu-\muo)\tau(\mu)d\mu\right)\\
    &= (1-\al)\left(\int_0^{\muh} (\muo-\mu)\tau(\mu)d\mu+\int_{\muh}^1 (\mu-\muu)\tau(\mu)d\mu\right)\\
    &\geq (1-\al)\left(\int_0^{\muh} (\muh-\mu)\tau(\mu)d\mu+\int_{\muh}^1 (\mu-\muh)\tau(\mu)d\mu\right),
\end{align*}
where the inequalities hold because $\muu\leq\muh(\muu,\muo)\leq\muo$ and the equality is a consequence of (\ref{eq:cond_1}) and (\ref{eq:cond_2}). Because $\tau$ has full support, the multipliers on both sides of the inequality are strictly positive, and thus $\al\geq 1-\al$, i.e., $\al\geq1/2$.

Now we argue that for $\al\in[1/2,1)$ the solution to (\ref{eq:cond_1}) and (\ref{eq:cond_2}) such that $\muu\leq\mu_0\leq\muo$ exists. Note that at $\al=1/2$, $[\muu,\muo]=[\mu_0,\mu_0]$ is a solution.

For the rest of this proof, we omit the dependence of $\Psi_i$ on $\al$. By \autoref{lem:binary-misaligned-adviser} and direct inspection, $\muh(\muu,\muo)$ is strictly and continuously increasing in its arguments, so $\Psi_1(\muu,\muo)$ is strictly and continuously decreasing in $\muu$ for each $\muo$. Furthermore, for every $\muo\in[\mu_0,1]$,
\begin{align*}
    \Psi_1(0,\muo)&=(1-\al)\int_{\muh(0,\muo)}^1\mu \tau(\mu)d\mu>0,\\
    \Psi_1(\mu_0,\muo)&\leq(1-2\al)\int_{\mu_0}^1(\mu-\mu_0)\tau(\mu)d\mu\leq0,
\end{align*}
where the second inequality follows from $\muh(\mu_0,\muo)\geq\mu_0$ and Bayes plausibility.
Therefore, for each $\muo\in[\mu_0,1]$, a best-response $b_1(\muo)\in[0,\mu_0]$ such that $\Psi_1(b_1(\muo),\muo)=0$ exists and is unique. Since $\Psi_1(\muu,\muo)$ strictly decreases in $\muo$, $b_1(\muo)$ strictly decreases in $\muo$.

Analogously, for each $\muu\in[0,\mu_0]$, a best-response $b_2(\muu)\in[\mu_0,1]$ such that $\Psi_2(\muu,b_2)=0$ exists, is unique, and strictly decreases in $\muu$.

Therefore, a solution to (\ref{eq:cond_1}) and (\ref{eq:cond_2}) is any $\muu\in[0,\mu_0]$ and $\muo=b_2(\muu)\in[\mu_0,1]$ such that $b_1(b_2(\muu))=\muu$. By the established properties of $b_1$ and $b_2$, $b_1(b_2(\muu))$ is continuous  in $\muu$ with $b_1(b_2(\muu))\in[0,\mu_0]$ for all $\muu\in[0,\mu_0]$; hence, $b_1(b_2(0))-0\geq 0$ and $b_1(b_2(\mu_0))-\mu_0\leq 0$. By the intermediate value theorem, there exists $\muu\in[0,\mu_0]$ such that $b_1(b_2(\muu))=\muu$.

So far, we showed that for $\al\geq1/2$, a solution exists and belongs to a closed rectangular set $D=\{(\muu,\muo):\muu\in[0,\mu_0],\muo\in[\mu_0,1]\}$. To establish uniqueness, consider the function $-\Psi=(-\Psi_1,-\Psi_2)$ on $D$. Any solution must satisfy $-\Psi(\muu,\muo)=(0,0)$. Observe that for any $(\muu,\muo)\in D$,
\begin{align*}
    \frac{\partial [-\Psi_1]}{\partial \muu}&=\al\int_0^{\muu}\tau(\mu)d\mu+(1-\al)\frac{\partial\muh}{\partial\muu}\tau(\muh)(\muh-\muu)+(1-\al)\int_{\muh}^1 \tau(\mu)d\mu>0,\\
    \frac{\partial [-\Psi_1]}{\partial \muo}&=(1-\al)\frac{\partial\muh}{\partial\muo}\tau(\muh)(\muh-\muu)\geq 0,\\
    \frac{\partial [-\Psi_2]}{\partial \muu}&=-(1-\al)\frac{\partial\muh}{\partial\muu}\tau(\muh)(\muh-\muo)\geq 0,\\
    \frac{\partial [-\Psi_2]}{\partial \muo}&=\al\int_{\muo}^1 \tau(\mu)d\mu+(1-\al)\frac{\partial\muh}{\partial\muo}\tau(\muh)(\muo-\muh)+(1-\al)\int_{0}^{\muh} \tau(\mu)d\mu>0.
\end{align*}
Moreover, for all $(\muu,\muo)\in D$, the Jacobian of $[-\Psi]$ is a P-matrix, i.e., it has strictly positive principal minors:
\begin{align*}
    \frac{\partial [-\Psi_1]}{\partial \muu}>0,\quad \frac{\partial [-\Psi_1]}{\partial \muu}\frac{\partial [-\Psi_2]}{\partial \muo}-\frac{\partial [-\Psi_1]}{\partial \muo} \frac{\partial [-\Psi_2]}{\partial \muu}>0.
\end{align*}
By the Gale-Nikaido Theorem  (\cite{GaleNikaido1965Jacobian}, Theorem 4), it follows that $[-\Psi]$ is injective on $D$, and thus there exists at most one solution to the equation $-\Psi(\muu,\muo)=(0,0)$.

It remains to establish that the constructed TRS is optimal and robustly rationalizable. By \autoref{thm:trust}, an optimal TRS exists. For $\al\in(1/2,1)$, \autoref{thm:viable} implies that its trust region is nondegenerate. The expressions for the partial derivatives of $V$ show that an optimal interval must also satisfy $0<\muu<\muo<1$. Moreover, the optimal interval contains $\mu_0$, as follows from \autoref{thm:minimax} and Bayes plausibility. Hence its endpoints must satisfy the first-order conditions in the domain $D$. Since these conditions have a unique solution in $D$, the proposed interval is optimal. Robust rationalizability then follows from \autoref{thm:minimax}.

\subsection{Proof of \autoref{prop:trust-change}}\label{app:proof-binary-state-change}
 
At $\al=1/2$, $[\muu,\muo]=[\mu_0,\mu_0]$ satisfies  conditions (\ref{eq:cond_1}) and (\ref{eq:cond_2}), using the continuous extension of $\muh$. At $\al=1$, full trust is optimal, so $[\muu,\muo]=[0,1]$.

For $\al\in(1/2,1)$, denote by $\Psi_{i1}$, $\Psi_{i2}$, and $\Psi_{i\al}$ the partial derivatives of $\Psi_i$ with respect to  $\muu$,  $\muo$, and $\al$ respectively. Define the Jacobian: 
\begin{align*}
    J(\muu,\muo,\al)\triangleq 
    \begin{pmatrix}
 \Psi_{11} & \Psi_{12}\\
\Psi_{21} & \Psi_{22}
\end{pmatrix}.
\end{align*}
As we argued in the proof of \autoref{prop:binary-trust-region}, for all $\muu$, $\muo$, $\al>1/2$, $\det J(\muu,\muo,\al)>0$, and therefore, by the implicit function theorem, optimal $\muu(\al)$ and $\muo(\al)$ are continuously differentiable and\footnote{Differentiating the optimality conditions with respect to $\al$ we obtain $\Psi_{11}\frac{d \muu}{d\al}+\Psi_{12}\frac{d \muo}{d\al}+\Psi_{1\al}=0, \Psi_{21}\frac{d \muu}{d\al}+\Psi_{22}\frac{d \muo}{d\al}+\Psi_{2\al}=0.$}
\begin{align*}
\begin{pmatrix} 
d\muu/d\al\\
d\muo/d\al 
\end{pmatrix}
= -J(\muu,\muo,\al)^{-1}
\begin{pmatrix}
\Psi_{1\al}\\ 
\Psi_{2\al}
\end{pmatrix}.
\end{align*}
Consequently,
\begin{align*}
    \frac{d \muu}{d\al}&=-\frac{\Psi_{2\muo}\Psi_{1\al}-\Psi_{1\muo}\Psi_{2\al}}{\Psi_{1\muu}\Psi_{2\muo}-\Psi_{1\muo}\Psi_{2\muu}}<0,\\
    \frac{d \muo}{d\al}&=\frac{\Psi_{2\muu}\Psi_{1\al}-\Psi_{1\muu}\Psi_{2\al}}{\Psi_{1\muu}\Psi_{2\muo}-\Psi_{1\muo}\Psi_{2\muu}}>0,
\end{align*}
where the inequalities hold because, as we already showed, $\Psi_{1\muu}<0$, $\Psi_{1\muo}\leq 0$, $\Psi_{2\muu}\leq0$, $\Psi_{2\muo}< 0$, and 
\begin{align*}
    \Psi_{1\al}= \int_0^{\muu}(\mu-\muu) \tau(\mu)d\mu -\int_{\muh}^1 (\mu-\muu) \tau(\mu)d\mu<0,\\
    \Psi_{2\al}= 
 \int_{\muo}^{1}(\mu-\muo) \tau(\mu)d\mu-\int_0^{\muh} (\mu-\muo) \tau(\mu)d\mu >0.
\end{align*}
The result follows.

\subsection{Proof of \autoref{prop:sensitive-change}}\label{app:proof-sensitive-change}

Throughout, fix $\al\in(1/2,1)$; the cases $\al\leq1/2$ and $\al=1$ are immediate. We begin with a simple lemma. 

\begin{lemma}\label{lemma_lr}
Let $U_1,U_2$ be twice differentiable and strictly convex on $[0,1]$. 
Assume that the ratio 
${U_1''(\mu)}/{U_2''(\mu)}$
is decreasing. Then, with $\muh_i$ defined for each $U_i$, $i\in \{1,\,2\}$, by equation \eqref{def_muh}, we have that $
\muh_1(\muu,\muo)\leq\muh_2(\muu,\muo)$ for all $\muu\leq\muo$.
\end{lemma}
\begin{proof}
By equation  \eqref{def_muh}, for $i\in \{1, 2\}$, $\muh_i(\muu,\muo)=\mathbb{E}_{\mu\sim f_i}[\mu]$,
where $f_{i}(\mu)$ is a density of a probability measure on $[\muu,\,\muo]$ defined as
\begin{align*}
f_{i}(\mu)\triangleq\frac{U_i^{\prime\prime}(\mu)}{\int_{\muu}^{\muo}U_i^{\prime\prime}(\mu)\,d\mu}.
\end{align*}
By assumption, ${f_{1}(\mu)}/{f_{2}(\mu)}$
is decreasing, and hence $f_2$ likelihood-ratio dominates $f_1$. This implies that the probability distribution $f_1$ is first-order stochastically dominated by $f_2$; in particular, it has a lower mean. The result follows.
\end{proof}
We now take the second step by showing a monotone relationship between the cutoff function $\muh$ and the trust region.  

\begin{lemma}\label{lem:cutoff-shift}
Consider two decision problems $U_i$, $i\in\{1,2\}$, and 
let $(\muu_i,\muo_i)\in D$ denote the  unique  solution to the system
\begin{align*}
\Psi^{i}_1(\muu,\muo,\alpha)=0,\qquad \Psi^{i}_2(\muu,\muo,\alpha)=0,
\end{align*}
where $\Psi^{i}_1,\Psi^{i}_2$ and $D$ are defined as in the proof of \autoref{prop:binary-trust-region} for each $U_i$.
If $\muh_1(\muu,\muo) \leq  \muh_2(\muu,\muo)$  for all $\muu\leq \muo$, 
then $\muu_2\leq \muu_1$ and $\muo_2\leq \muo_1$.
\end{lemma}

\begin{proof}
For any $h\in[0,1]$, define the auxiliary functions
\begin{align*}
\widetilde\Psi_1(\muu,h)
&\triangleq \alpha \int_0^{\muu} (\mu-\muu)\tau(\mu)\,d\mu
      + (1-\alpha)\int_h^1 (\mu-\muu)\tau(\mu)\,d\mu,\\
\widetilde\Psi_2(\muo,h)
&\triangleq \alpha \int_{\muo}^1 (\mu-\muo)\tau(\mu)\,d\mu
      + (1-\alpha)\int_0^h (\mu-\muo)\tau(\mu)\,d\mu.
\end{align*}

For each $h$, let $\muu(h)$ be the unique solution to $\widetilde\Psi_1(\muu,h)=0$ in $[0,\mu_0]$,
and let $\muo(h)$ be the unique solution to $\widetilde\Psi_2(\muo,h)=0$ in $[\mu_0,1]$, with the 
 existence and uniqueness following from an argument analogous to that used in the proof of \autoref{prop:binary-trust-region}.

For each $i\in\{1,2\}$ define the scalar mapping
\begin{align*}
\varphi_i(h)\ \triangleq\ \muh_i\big(\muu(h),\muo(h)\big).
\end{align*}
Let $h_i\triangleq \muh_i(\muu_i,\muo_i)$ be the cutoff evaluated at the optimal endpoints of problem $i$.
Then $(\muu_i,\muo_i)$ solves $\Psi_1^i=\Psi_2^i=0$ if and only if $\muu_i=\muu(h_i)$, $\muo_i=\muo(h_i)$, and $h_i=\varphi_i(h_i)$.
By the assumption $\muh_1\leq \muh_2$ pointwise, for every $h$,
\begin{align*}
\varphi_1(h)=\muh_1(\muu(h),\muo(h))\ \le\
\muh_2(\muu(h),\muo(h))=\varphi_2(h).
\end{align*}

Define a region $H\triangleq\{h\in[0,1]:\muu(h)\leq h\leq\muo(h)\}$. Note that $H$ is an interval and $h_1,h_2\in H$. For $h\in H$,
implicit differentiation yields
\begin{align*}
\muu'(h)= -\frac{\partial \widetilde\Psi_1/\partial h}{\partial \widetilde\Psi_1/\partial \muu}\leq 0,
\qquad
\muo'(h)= -\frac{\partial \widetilde\Psi_2/\partial h}{\partial \widetilde\Psi_2/\partial \muo}\leq 0.
\end{align*}

Thus, $\varphi_i$ is weakly decreasing in $h$ on $H$: both $\muu(h)$ and $\muo(h)$ are weakly
decreasing in $h$, while $\muh_i(\muu,\muo)$ is weakly increasing in each endpoint, so the composition
$h\mapsto \varphi_i(h)$ is weakly decreasing.

It follows that $h_2\geq h_1$: if $h_2<h_1$, then, since $\varphi_2$ is decreasing on $H$,
\begin{align*}
h_2=\varphi_2(h_2)\geq \varphi_2(h_1)\geq \varphi_1(h_1)=h_1,
\end{align*}
which is a contradiction. Since $\muu(\cdot)$ and $\muo(\cdot)$ are weakly decreasing,
\begin{align*}
\muu_2=\muu(h_2)\leq \muu(h_1)=\muu_1,
\qquad
\muo_2=\muo(h_2)\leq \muo(h_1)=\muo_1,
\end{align*}
completing the proof.
\end{proof}

By \autoref{lemma_lr}, $U_1''/U_2''$ decreasing implies $\muh_1(\muu,\muo)\le \muh_2(\muu,\muo)$ for all $\muu\le\muo$.  \autoref{lem:cutoff-shift} then yields $\muu_2\le \muu_1$ and $\muo_2\le \muo_1$. This proves \autoref{prop:sensitive-change}.

\section{Additional results}\label{app:additional}

\subsection{Sufficient Conditions for Regularity}\label{app:convex-utility}
In this appendix, we show that the agent's decision problem is regular when her private information induces a full-support distribution of beliefs with a density, under additional genericity assumptions.

We assume that the action space is finite and the agent's ex-post payoff is
type-independent, $u(a,\omega)$. It is convenient to identify the agent's type $\theta$ with the posterior belief it induces under the reference prior $\mu_0$. Thus,
$\Theta=\Delta(\Omega)$ and
$\theta(\omega)$ is the posterior probability of state $\omega$.
For an arbitrary initial belief $\mu$, the agent's posterior
after observing $\theta$ is
\begin{align}
\nu_{\mu,\theta}(\omega)
\triangleq
\frac{\frac{\mu(\omega)}{\mu_0(\omega)}\theta(\omega)}
{\sum_{\omega'\in\Omega}
\frac{\mu(\omega')}{\mu_0(\omega')}\theta(\omega')}.
\end{align}

A necessary and sufficient condition for a private strategy $\sigmah$ to be Bayes-optimal at belief $\mu$ is that, for almost every $\theta$ under the type distribution induced by $\mu$, $\sigmah(\cdot|\theta)\in\Delta(A)$ is an optimal best-response  with respect to $\nu_{\mu,\theta}$: for all $a\in\supp\,\sigmah(\cdot|\theta)$,
\begin{align}
a\in\arg\max_{a'\in A}\sum_{\omega\in\Omega} \nu_{\mu,\theta}(\omega)u(a',\omega).
\end{align}

\begin{assumption}\label{ass:convex-genericity}
Any two distinct actions yield different payoffs in every state:
$u(a,\omega)\neq u(a',\omega)$ for all $a\neq a'$ and
$\omega\in\Omega$.
Moreover, there exist two distinct actions $a_1,a_2\in A$
and an interior belief $\mu\in\Delta(\Omega)$ such that
$\expect_\mu[u(a_1,\omega)]
=\expect_\mu[u(a_2,\omega)]
>\expect_\mu[u(a,\omega)]$,  for all $a\notin\{a_1,a_2\}$.
\end{assumption}

The first condition is restrictive but holds generically for finite state and action spaces. It ensures that, for any belief, the agent's optimal action is unique for almost every private type. The second condition requires that exactly two actions be optimal at some interior belief. Together with the first condition, it ensures that any change in the agent's prior shifts the boundary between private types who prefer one action over the other. With sufficiently rich private information, this implies that any change in beliefs is decision-relevant for some agent types.

\begin{lemma}
Suppose $\theta$ has a full-support distribution on $\Delta(\Omega)$ with a density with respect to Lebesgue measure under the reference prior $\mu_0$, and \autoref{ass:convex-genericity} holds. Then the agent's decision problem is regular.
\end{lemma}
\begin{proof}
We first establish part (i) of regularity. Let $\bar f$
denote the distribution of $\theta$ under $\mu_0$.
Since $\theta$ is the agent's posterior under $\mu_0$,
Bayes' rule implies
\[
f(d\theta\mid\omega)
=\frac{\theta(\omega)}{\mu_0(\omega)}\,\bar f(d\theta).
\]
Since $\theta$ belongs to the interior of the simplex
almost surely, the distributions of types under $\mu$
and under $\mu_0$ are mutually absolutely continuous,
for any belief $\mu$. Thus, when referring to events
that occur with zero or full probability, we need not
specify the belief under which probabilities are evaluated.

Fix $\mu\in\Delta(\Omega)$ and two distinct actions
$a,a'\in A$. These actions can be simultaneously optimal
only if
\[
\sum_{\omega\in\Omega}
\frac{\mu(\omega)}{\mu_0(\omega)}\,\theta(\omega)
\bigl[u(a,\omega)-u(a',\omega)\bigr]=0.
\]
By \autoref{ass:convex-genericity}, the actions have
different payoffs in every state. Thus, this equality
defines a proper hyperplane in the belief simplex
(or an empty set), which has measure zero since
$\theta$ has a density. As $A$ is finite, the optimal
action is therefore unique for almost every type,
under any belief $\mu$. Since the state-conditional
distributions of types are absolutely continuous
with respect to $\bar f$, all optimal decision rules
induce the same profile of state-contingent expected
payoffs.

We next establish part (ii). Fix any distinct interior
beliefs $\mu_1,\mu_2$, and let $a_1,a_2$ and
$\mu^*$ be as in \autoref{ass:convex-genericity}.
Define
\[
d(\omega)\triangleq u(a_1,\omega)-u(a_2,\omega),
\qquad
L_i(\theta)\triangleq
\sum_{\omega\in\Omega}
\frac{\mu_i(\omega)}{\mu_0(\omega)}
\theta(\omega)d(\omega),
\quad i=1,2.
\]
There exists an interior type
$\theta^*$ whose posterior under $\mu_1$ is $\mu^*$.
Hence $L_1(\theta^*)=0$, and $a_1,a_2$ strictly
outperform all other actions at $\theta^*$ under
$\mu_1$. By continuity, the latter property
continues to hold in some open neighborhood
$O$ of $\theta^*$.

Since $\mu_1\neq\mu_2$ and $d(\omega)\neq 0$
for every $\omega$, the hyperplanes
$\{\theta:L_1(\theta)=0\}$ and
$\{\theta:L_2(\theta)=0\}$ are distinct.
Consequently, there exists $\bar\theta\in O$
such that
$L_1(\bar\theta)=0$ and
$L_2(\bar\theta)\neq 0$.
Without loss of generality, suppose
$L_2(\bar\theta)<0$; otherwise interchange
$a_1$ and $a_2$. By continuity, there exists a nonempty open
set $B\subset O$ such that
\[
L_1(\theta)>0
\quad\text{and}\quad
L_2(\theta)<0
\qquad\text{for every }\theta\in B.
\]
For every $\theta\in B$, action $a_1$ is
uniquely optimal under $\mu_1$, but is not
optimal under $\mu_2$. Since $B$ has positive
probability under both beliefs, no decision
rule can be Bayes-optimal at both $\mu_1$
and $\mu_2$. This proves part (ii).
\end{proof}

\subsection{Analytical solutions with multiple states}

In this appendix, we present solutions to two tractable cases within the multi-state environment. First, we provide a complete characterization for the case of binary actions and no private information of the agent. Then, we present a solution to a parametric class of problems in which symmetry permits an explicit characterization of the optimal trust region. 

\subsubsection{Binary Action}

Consider the setting in which the agent has only two pure private strategies, i.e., $A = \{a_1, a_2\}$ and $|\Theta|=1$. Thus, the agent has no private information and we drop the type throughout. 

Without loss of generality, we can normalize the agent's payoff from action $a_1$ to zero, $u(a_1,\omega) \equiv 0$, and denote the expected payoff from action $a_2$ when the adviser's posterior is $\mu$ by $v(\mu)\triangleq\expect_{\mu} [u(a_2,\omega)]$. Denote by $\tauv \in \Delta(\reals)$ the distribution of $v$ when $\mu$ is distributed according to $\tau\in\Delta(\Delta(\Omega))$.

Define the absolute losses and gains from taking the second action relative to the first one:
\begin{align}
    L(\tauv)= \int_{-\infty}^{0} (-v) \, \tauv(dv), \quad G(\tauv)= \int_{0}^{+\infty} v \, \tauv(dv).
\end{align}
Also, define the following threshold:
\begin{align}\label{eq:trust-region-two-action}
\alh(\tauv)=\frac{\max\{L(\tauv),G(\tauv)\}}{L(\tauv)+G(\tauv)}.
\end{align}
To rule out trivial cases and to simplify the exposition of the optimal strategy, we make the following assumption that holds in generic non-degenerate environments. 
\begin{assumption}[Generic Nondegeneracy]\label{ass:binary-action-gener}
    $L(\tauv) > 0$, $G(\tauv) > 0$, $L(\tauv) \neq G(\tauv)$, $\tau(\{\mu : v(\mu) = 0\}) = 0$.
\end{assumption}

\begin{proposition}[Binary Action Solution]\label{prop:coarse-solution}
Suppose that \autoref{ass:binary-action-gener} holds. If $\al\neq \alh(\tauv)$, then the optimal solution exists, is unique, and is robustly rationalizable. The trust region can be chosen as follows: if $\al>\alh(\tauv)$, then all messages are trusted, $T=\Delta(\Omega)$; if $\al<\alh(\tauv)$, then no messages are trusted, $T=\{\mu_0\}$. If $\al=\alh(\tauv)$, both full trust and no trust are optimal and robustly rationalizable.

\end{proposition}

By \autoref{prop:coarse-solution}, generically, the optimal solution is stark: either all or none of the adviser's messages are trusted. This is in contrast to the binary-state case with a rich strategy space, where the trust region expanded continuously with the alignment probability (\autoref{prop:trust-change}). 

Notably, only the aggregate quantities $L(\tauv)$ and $G(\tauv)$ matter for the determination of the trust region; the detailed distribution of relative payoffs $\tauv$ is irrelevant. The threshold $\alh(\tauv)$ is minimized at $L(\tauv)=G(\tauv)$, in which case $\alh(\tauv)=1/2$. Hence, if $\al<1/2$, the agent never trusts the adviser, regardless of $\tauv$.\footnote{This conclusion is driven by the coarseness of the strategy space; we know from  \autoref{thm:viable} that, for any $\al>0$, the agent would optimally use a non-trivial trust region if the state space and her action space were sufficiently rich.} By contrast, for a given $\al>1/2$, the trust condition $\alh(\tauv)\leq\al$ is equivalent to $(L(\tauv),G(\tauv))$ lying in the cone defined by the two linear inequalities
\begin{align*}
(1-\al)L(\tauv) \leq \al G(\tauv),\quad (1-\al)G(\tauv) \leq \al L(\tauv).
\end{align*}
Thus, in binary decision problems, the adviser is beneficial only when the expected gains and losses of one action relative to the other are not too far apart.
 
\paragraph{Proof of \autoref{prop:coarse-solution}.}\label{app:proof-binary-action}
With a small abuse of notation, we can parameterize each private strategy by $\sigmah=\Pr(a=a_2)$. We also drop the dependence of $G$ and $L$ on $\tauv$ in the notation. Then, by the arguments behind \autoref{thm:trust}, if the agent employs the set of private strategies $\hat\Sigma_0=\{\sigmah(m)\}_{m\in \Delta(\Omega)}$, the payoffs coming from both aligned and misaligned adviser depend only on $\sigmah_L\triangleq\inf \hat\Sigma_0$ and  $\sigmah_H\triangleq\sup \hat\Sigma_0$, and the optimal payoffs from using $\hat\Sigma_0$ are the same as if the agent  plays $\sigmah(m)=\sigmah_L$ when $v(m)<0$ and $\sigmah(m)=\sigmah_H$ when $v(m)\geq 0$. This payoff is:
\begin{align}
    &\int_{-\infty}^{0} (\al \sigmah_L+(1-\al) \sigmah_H) v\, \tauv(dv) + \int_{0}^{+\infty} (\al \sigmah_H+(1-\al) \sigmah_L) v\, \tauv(dv)\notag\\
    &=\sigmah_L((1-\al)G-\al L)+\sigmah_H(\al G-(1-\al) L)\label{eq:binary_action_payoff}.
\end{align}
The optimal choice of $\sigmah_L$ and $\sigmah_H$ must maximize \eqref{eq:binary_action_payoff} subject to $\sigmah_L,\sigmah_H\in[0,1]$ and $\sigmah_L\leq\sigmah_H$.
This is a linear optimization subject to $(\sigmah_L,\sigmah_H)$ being in a triangle  with vertices $(0,0)$, $(0,1)$, and $(1,1)$. A straightforward calculation gives the following solution:

If $G=L$: if $\al<\alh\triangleq 1/2$, then any $\sigmah_L=\sigmah_H$ is optimal; if $\al>\alh$, then $\sigmah_L=0$ and $\sigmah_H=1$; if $\al=\alh$, then any $(\sigmah_L,\sigmah_H)$ is optimal.
If $G>L$: if $\al<\alh$, then $\sigmah_L=\sigmah_H=1$; if $\al>\alh$, then $\sigmah_L=0$ and $\sigmah_H=1$; if $\al=\alh$, then $\sigmah_H=1$ and any $\sigmah_L$ is optimal.
If $G<L$: if $\al<\alh$, then $\sigmah_L=\sigmah_H=0$; if $\al>\alh$, then $\sigmah_L=0$ and $\sigmah_H=1$; if $\al=\alh$, then $\sigmah_L=0$ and any $\sigmah_H$ is optimal.

The cases $\sigmah_L=\sigmah_H$ correspond to not trusting any message and always acting in the same way, optimal at the prior belief, so $T=\{\mu_0\}$. The cases $\sigmah_L=0$ and $\sigmah_H=1$  correspond to trusting all messages, so $T=\Delta(\Omega)$.  Since we assumed that the probability of $v(\mu)=0$ is $0$ and $G\neq L$, the corresponding optimal strategy is uniquely determined.

It is left to show that those strategies are robustly rationalizable. Define $M_0=\{\mu:v(\mu)=0\}$, $M_-=\{\mu:v(\mu)<0\}$, and $M_+=\{\mu:v(\mu)>0\}$. Define probability measures $q_+(X)=\int_X v(\mu)\tau(d\mu)/G$ for $X\subseteq M_+$, $q_-(Y)=\int_Y (-v(\mu))\tau(d\mu)/L$ for $Y\subseteq M_-$.

For $\al>\alh$, the agent fully trusts the adviser. Consider the following strategy of the misaligned adviser. If $\mu\in M_0$, then $\beta(\mu)=\mu$. If $\mu\in M_-$, then $\beta$ randomizes over messages $m\in M_+$ according to $q_+$. If $\mu\in M_+$, then $\beta$ randomizes over messages $m\in M_-$ according to $q_-$. This strategy is clearly adversarial. Furthermore, since $\al>\alh$,  after any message $m\in M_+$, the posterior expected payoff from action $a_2$ is strictly positive: for any $X\subseteq M_+$  with $\tau(X)>0$,
\begin{align*}
  \al\int_X v(m)\tau(dm)+(1-\al)\int_{\Delta(\Omega)} v(\mu)\beta(X|\mu)\tau(d\mu)=\al G q_+(X)-(1-\al)Lq_+(X)>0.
\end{align*}
Analogously, after any message $m\in M_-$, the posterior expected payoff from action $a_2$ is strictly negative:
for any $Y\subseteq M_-$ with $\tau(Y)>0$,
\begin{align*}
  \al\int_Y v(m)\tau(dm)+(1-\al)\int_{\Delta(\Omega)} v(\mu)\beta(Y|\mu)\tau(d\mu)=-\al L q_-(Y)+(1-\al)Gq_-(Y)<0.
\end{align*}
And by construction, after any message $m\in M_0$, the posterior expected payoff from action $a_2$ is zero. 

For $\al<\alh$, the agent doesn't trust the adviser so any adviser's strategy is adversarial. Consider the case $G>L$; the complementary case is analogous. We need to construct the misaligned adviser strategy that makes communication not valuable. To do so, define $\gamma=\al L/((1-\al)G)\in[0,1]$. Consider the following strategy of the misaligned adviser. If $\mu\in M_-$,   then $\beta$ randomizes over messages $m\in M_+$ according to $q_+$. If $\mu\in M_+$, then with probability $\gamma$, $\beta$ randomizes over messages $m\in M_-$ according to $q_-$ and with probability $(1-\gamma)$, $\beta$ randomizes over messages $m\in M_+$ according to $q_+$. This strategy makes the posterior expected payoff from action $a_2$ zero after every message $m\in M_-$: for any $Y\subseteq M_-$ with $\tau(Y)>0$,
\begin{align*}
  \al\int_Y v(m)\tau(dm)+(1-\al)\int_{\Delta(\Omega)} v(\mu)\beta(Y|\mu)\tau(d\mu)=-\al L q_-(Y)+(1-\al)\gamma Gq_-(Y)=0.
\end{align*}
After any message $m\in M_+$, the posterior expected payoff from action $a_2$ is strictly positive: for any $X\subseteq M_+$ with $\tau(X)>0$,
\begin{align*}
  \al\int_X v(m)\tau(dm)+(1-\al)\int_{\Delta(\Omega)} v(\mu)\beta(X|\mu)\tau(d\mu)=(G-L)q_+(X)>0.
\end{align*}
Thus, this $\beta$ robustly rationalizes the strategy. 

Finally, at $\al=\alh$, both full trust and no trust are optimal, along with a continuum of other strategies, and robustly rationalizable by the same strategies $\beta$ as above. 

\subsubsection{Spherical Environments}\label{app:spherical}

In this appendix, we characterize the optimal trust region in radially symmetric environments.

Use coordinates centered at the prior $\mu_0$, so beliefs lie in the translated simplex $\Delta(\Omega)-\mu_0$ and the origin represents the prior. With a slight abuse of notation, write $\mu$ for the centered belief and $U$ for the translated indirect utility. Let $U(\mu)=\val(\|\mu\|)$ be twice continuously differentiable, with $\val''>0$. Let the adviser's belief be symmetrically distributed over a ball $C=\{\mu:\|\mu\|\leq \rad_0\}$ in the affine hull of the translated simplex, contained in $\Delta(\Omega)-\mu_0$, with radial density $\tau(\rad)$. Then, there exists a robustly rationalizable solution in which the trust region is the ball $T=\{\mu:\|\mu\|\leq \rad^*(\al)\}$.

We will show this result via \autoref{thm:minimax} by explicitly constructing the corresponding TRE. The key observation, that we formalize and prove in  \autoref{lem:spherical}, is that the misaligned adviser with belief $\mu$ induces an antipodal belief on the boundary of $T$. This fact combined with the symmetry of the problem implies that the adversarial strategy is the same on each ray going through the center of the ball, and hence analogous to the strategy constructed in \autoref{sec:one-dim}.

The radius $\rad^*(\al)$ can be found via the balancing condition applied to any ray going through the center. Indeed, consider any line through the origin, with signed coordinate $\rad$ and boundary points $-\rad_0$ and $\rad_0$. Then, the belief at $\rad=\rad^*(\al)$ will be induced by the misaligned adviser only when his belief is at negative coordinates, and by the aligned adviser only when his belief is at positive coordinates. Thus, the agent's interim beliefs satisfy the TRE property if and only if:
\begin{align*}
    \rad^*=\frac{\al\int_{\rad^*}^{\rad_0}\rad \tau(\rad)d\rad-(1-\al)\int_{0}^{\rad_0}\rad \tau(\rad)d\rad}{\al\int_{\rad^*}^{\rad_0} \tau(\rad)d\rad+(1-\al)\int_{0}^{\rad_0} \tau(\rad)d\rad}.
\end{align*}
Rearranging, we obtain:
\begin{align}
(2\al-1)\int_{\rad^*}^{\rad_0}(\rad-\rad^*)\tau(\rad)d\rad=(1-\al)\left(\int_0^{\rad^*}(\rad+\rad^*)\tau(\rad)d\rad+\int_{\rad^*}^{\rad_0}2\rad^* \tau(\rad)d\rad\right).   
\end{align}
For $\al<1/2$, the equation does not admit a solution. In this case, the optimal radius is $\rad^*(\al)=0$, since the misaligned adviser can make communication uninformative by reporting the antipodal belief with probability $\al/(1-\al)$ and an independent draw from $\tau$ otherwise. At $\al=1/2$, $\rad^*=0$ is a solution. For $\al\in(1/2,1)$,
the left-hand side is continuously and strictly decreasing in $\rad^*$, and the right-hand side is continuously and strictly increasing in $\rad^*$, with a derivative with respect to $\rad^*$ that is strictly positive. Therefore, the equation admits a unique solution $\rad^*(\al)$. As the left-hand side strictly increases in $\al$ and the right-hand side strictly decreases in $\al$, $\rad^*(\al)$ strictly increases in $\al$. Furthermore, $\rad^*(1)=\rad_0$.

This example features two notable properties. First,  the $\MVA$ does not depend on $U$ or the number of states; the threshold alignment probability is always $1/2$.
Second, the shape of $\val$, which captures the details of the decision problem, does not matter for the trust region $T$; the trust region is uniquely pinned down by $\tau(\rad)$. For example, if $\tau$ is uniform, $\tau\sim U[0,\rad_0]$, then, for $\alpha\geq1/2$, $
\rad^*(\al)=\frac{1-\sqrt{1+\al-2\al^2}}{\al}\rad_0$.

\paragraph{Supporting calculations.}

The key step in the construction of the optimal trust region comes from the following simple lemma. 

\begin{lemma}[Spherical U]\label{lem:spherical}
Let $U(\mu)=\val(\|\mu\|)$ be as above and $\mu'(\rad,n)=\rad n$, where $n$ is a unit direction parallel to the belief simplex and $\rad\in[0,\rad_0]$. Then, for any $\mu\in C$, (i) $D_U(\mu,\mu'(\rad,n))$ is strictly increasing in $-n\cdot\mu$ whenever $\rad>0$, and (ii) $D_U(\mu,\mu'(\rad,n))$ is a unimodal function in $\rad$ with a minimum at $\rad=\max\{0,n\cdot\mu\}$. 
\end{lemma}
\begin{proof}
In this case, $\nabla U(\mu'(\rad,n))=\val'(\rad)n$, and thus
\begin{align*}
    D_U(\mu,\mu'(\rad,n))=\val(\|\mu\|)-\val(\rad)-\val'(\rad)(n\cdot\mu-\rad).
\end{align*}
By our convexity assumptions on $U$, $\val'(\rad)>0$ whenever $\rad>0$ and $\val''(\rad)>0$. Thus $D_U(\mu,\mu'(\rad,n))$ is strictly increasing in $-n\cdot\mu$ for $\rad>0$. Furthermore,
\begin{align*}
    \frac{\partial D_U(\mu,\mu'(\rad,n))}{\partial \rad}=\val''(\rad)(\rad-n\cdot\mu).
\end{align*}
Because $\val''(\rad)>0$, the unimodality follows.
\end{proof}

Suppose that the misaligned adviser holds belief $\mu$. The adviser will try to maximize $D_U(\mu,\mu')$ over $\mu'$ in the trust region. For $\mu\neq0$, parameterizing $\mu'=\rad n$ and applying \autoref{lem:spherical} yields that $n$ is optimally chosen to be $-\mu/\|\mu\|$, whereas $\rad$ is chosen to be maximal within the trust region, $\rad=\rad^*(\al)$. Thus, when $\rad^*(\al)>0$, the optimal $\mu'$ is uniquely given by $-\rad^*(\al)\mu/\|\mu\|$. If $\mu=0$ and $\rad^*(\al)>0$, every point on the boundary is optimal. If $\rad^*(\al)=0$, the unique optimal belief is $\mu'=0$.

\end{document}